\documentclass{svjour3}

\usepackage[english]{babel}
\usepackage[dvips, usenames]{color}
\color{black}

\usepackage[ruled, vlined, nofillcomment, linesnumbered]{algorithm2e}

\usepackage{pbox}
\usepackage{cite}
\usepackage{url}
\usepackage{braket}
\usepackage{amsmath}
\usepackage{amssymb}
\usepackage{graphicx}
\usepackage{subfigure}
\usepackage{booktabs}
\usepackage{xspace}
\usepackage{tikz}
\usetikzlibrary{snakes,arrows,shapes,positioning,fit,calc}
\usepackage{pgfplots}

\usepackage[pdftitle={Mining Closed Strict Episodes},
            pdfauthor={Nikolaj Tatti, Boris Cule},
            pdfkeywords={Frequent Episode Mining, Closed Episodes,  Level-wise Algorithm},
            ocgcolorlinks = true, bookmarksopen = true, bookmarksnumbered = true, hyperfootnotes = false]{hyperref} 
\usepackage[all]{hypcap}

\hypersetup{linkcolor=[rgb]{0.949, 0.482, 0.216}} 
\hypersetup{citecolor=[rgb]{0.925, 0.165, 0.224}} 
\hypersetup{urlcolor=[rgb]{0.4, 0.165, 0.553}}

\usepackage{ifthen}
\newcommand{\pr}[1]{\left(#1\right)}
\newcommand{\fpr}[1]{\mathopen{}\left(#1\right)}

\newcommand{\abs}[1]{{\left|#1\right|}}

\newcommand{\enset}[2]{\left\{#1 ,\ldots , #2\right\}}
\newcommand{\enpr}[2]{\pr{#1 ,\ldots , #2}}

\newcommand{\funcdef}[3]{{#1}:{#2} \to {#3}}

\newcommand{\p}{\textbf{P}}

\newcommand{\define}{\leftarrow}

\newcommand{\dispfunc}[2]{%
  \ensuremath{%
  \ifthenelse{\equal{\noexpand#2}{}}%
    {{#1}}%
    {{#1}\fpr{#2}}}}

\newcommand{\closure}[1]{\mathit{tcl}\fpr{#1}}
\newcommand{\iclosure}[1]{\mathit{icl}\fpr{#1}}
\newcommand{\iclN}[1]{\mathit{icl}_N\fpr{#1}}
\newcommand{\iclE}[1]{\mathit{icl}_E\fpr{#1}}
\newcommand{\iclEN}[1]{\mathit{icl}_{\mathit{EN}}\fpr{#1}}

\newcommand{\last}[1]{\dispfunc{\mathit{last}}{#1}}

\newcommand{\lab}[1]{\dispfunc{\mathit{lab}}{#1}}
\newcommand{\map}[1]{\mathit{map}\fpr{#1}}

\newcommand{\efam}[1]{\mathcal{#1}}

\newcommand{\freq}[1]{\mathit{fr}\fpr{#1}}
\newcommand{\freqf}[1]{\mathit{fr_f}\fpr{#1}}
\newcommand{\freqm}[1]{\mathit{fr_d}\fpr{#1}}
\newcommand{\sspace}{\mathcal{S}}

\SetKwComment{tcpas}{\{}{\}}
\SetCommentSty{textnormal}
\SetArgSty{textnormal}
\SetKw{False}{false}
\SetKw{True}{true}
\SetKw{Null}{null}
\SetKw{Output}{output}
\SetKw{AND}{and}
\SetKw{OR}{or}
\SetAlgoSkip{}
\SetKwInOut{Input}{input}
\SetKwInOut{Out}{output}

\renewenvironment{proof}{\begin{oldproof}}{{\hfill \ensuremath{\Box}}\end{oldproof}}

\makeatletter 
\def\underbracket{%
 \@ifnextchar[{\@underbracket}{\@underbracket [\@bracketheight]}%
} 
\def\@underbracket[#1]{%
 \@ifnextchar[{\@under@bracket[#1]}{\@under@bracket[#1][0.4em]}%
} 
\def\@under@bracket[#1][#2]#3{
 \mathop{\vtop{\m@th \ialign {##\crcr $\hfil \displaystyle {#3}\hfil $%
 \crcr \noalign {\kern 3\p@ \nointerlineskip }\upbracketfill {#1}{#2} 
 \crcr \noalign {\kern 3\p@ }}}}\limits} 
 \def\upbracketfill#1#2{$\m@th \setbox \z@ \hbox {$\braceld$} 
 \edef\@bracketheight{\the\ht\z@}\bracketend{#1}{#2} 
 \leaders \vrule \@height #1 \@depth \z@ \hfill 
 \leaders \vrule \@height #1 \@depth \z@ \hfill \bracketend {#1}{#2}$} 
\def\bracketend#1#2{\vrule height #2 width #1\relax} 
\makeatother

\makeatletter 
\pgfdeclareshape{event}{ 
\inheritsavedanchors[from=rectangle] 
\inheritanchorborder[from=rectangle] 
\inheritanchor[from=rectangle]{center} 
\inheritanchor[from=rectangle]{north} 
\inheritanchor[from=rectangle]{south} 
\inheritanchor[from=rectangle]{south east} 
\inheritanchor[from=rectangle]{south west} 
\inheritanchor[from=rectangle]{west} 
\inheritanchor[from=rectangle]{east} 
\backgroundpath{
\southwest \pgf@xa=\pgf@x \pgf@ya=\pgf@y 
\northeast \pgf@xb=\pgf@x \pgf@yb=\pgf@y 
\pgf@yc=\pgf@ya \advance\pgf@yc by 2pt 
\pgfpathmoveto{\pgfpoint{\pgf@xa}{\pgf@yc}} 
\pgfpathlineto{\pgfpoint{\pgf@xa}{\pgf@ya}} 
\pgfpathlineto{\pgfpoint{\pgf@xb}{\pgf@ya}} 
\pgfpathlineto{\pgfpoint{\pgf@xb}{\pgf@yc}} 
} 
}
\makeatother 

\newcommand{\eventseq}[2]{%
\foreach \x [count=\xi from 1, count = \xprev from 0]  in {#1}{%
    \node[draw,shape=event, right=0 of #2\xprev.south east, anchor = south west, inner sep = 2pt, outer sep = 0pt] (#2\xi) {\x};%
}%
}

\newcommand{\eventwin}[6]{\filldraw[fill=#4, draw = #4] let \p1 = (#1.south west), \p2 = (#2.south east) in (\x1 + #5, \y1 - #3) rectangle (\x2 - #6, \y2 - #3 - 2pt);}

\pgfdeclarelayer{background}
\pgfdeclarelayer{foreground}
\pgfsetlayers{background,main,foreground}

\tikzstyle{block} = [rounded corners, draw=blue!70, fill=white, text width=3.3cm, minimum height=4em]
\tikzstyle{bgblock} = [rounded corners, draw=blue!70, thick, fill=blue!10, text width=3.3cm, minimum height=4em]
\tikzstyle{line} = [draw, -latex', thick,blue!70]

\definecolor{yafaxiscolor}{rgb}{0.3, 0.3, 0.3}

\definecolor{yafcolor1}{rgb}{0.4, 0.165, 0.553}
\definecolor{yafcolor2}{rgb}{0.949, 0.482, 0.216}
\definecolor{yafcolor3}{rgb}{0.47, 0.549, 0.306}
\definecolor{yafcolor4}{rgb}{0.925, 0.165, 0.224}
\definecolor{yafcolor5}{rgb}{0.141, 0.345, 0.643}
\definecolor{yafcolor6}{rgb}{0.965, 0.933, 0.267}
\definecolor{yafcolor7}{rgb}{0.627, 0.118, 0.165}
\definecolor{yafcolor8}{rgb}{0.878, 0.475, 0.686}

\newlength{\yafaxispad}
\newlength{\yaftlpad}
\newlength{\yaflabelpad}
\newlength{\yafaxiswidth}
\newlength{\yafticklen}
\makeatletter
\def\pgfplots@drawtickgridlines@INSTALLCLIP@onorientedsurf#1{}
\makeatother

\newcommand{\yafdrawaxis}[4]{
	\pgfplotstransformcoordinatex{#1}\let\xmincoord=\pgfmathresult 
	\pgfplotstransformcoordinatex{#2}\let\xmaxcoord=\pgfmathresult 
	\pgfplotstransformcoordinatey{#3}\let\ymincoord=\pgfmathresult 
	\pgfplotstransformcoordinatey{#4}\let\ymaxcoord=\pgfmathresult 
	\pgfsetlinewidth{\yafaxiswidth} 
	\pgfsetcolor{yafaxiscolor}
	\pgfpathmoveto{\pgfpointadd{\pgfpointadd{\pgfplotspointrelaxisxy{0}{0}}{\pgfqpointxy{\xmincoord}{0}}}{\pgfqpoint{-0.5\yafaxiswidth}{\yafaxispad}}}
	\pgfpathlineto{\pgfpointadd{\pgfpointadd{\pgfplotspointrelaxisxy{0}{0}}{\pgfqpointxy{\xmaxcoord}{0}}}{\pgfqpoint{0.5\yafaxiswidth}{\yafaxispad}}}
	\pgfpathmoveto{\pgfpointadd{\pgfpointadd{\pgfplotspointrelaxisxy{0}{0}}{\pgfqpointxy{0}{\ymincoord}}}{\pgfqpoint{\yafaxispad}{-0.5\yafaxiswidth}}}
	\pgfpathlineto{\pgfpointadd{\pgfpointadd{\pgfplotspointrelaxisxy{0}{0}}{\pgfqpointxy{0}{\ymaxcoord}}}{\pgfqpoint{\yafaxispad}{0.5\yafaxiswidth}}}
	\pgfusepath{stroke}
}

\pgfplotscreateplotcyclelist{yaf}{%
{yafcolor1,mark options={scale=0.75},mark=o}, 
{yafcolor2,mark options={scale=0.75},mark=square},
{yafcolor3,mark options={scale=0.75},mark=triangle},
{yafcolor4,mark options={scale=0.75},mark=o},
{yafcolor5,mark options={scale=0.75},mark=o},
{yafcolor6,mark options={scale=0.75},mark=o},
{yafcolor7,mark options={scale=0.75},mark=o},
{yafcolor8,mark options={scale=0.75},mark=o}} 

\pgfplotsset{axis y line=left, axis x line=bottom,
	tick align=outside,
	tickwidth=\yafticklen,
	clip = false,
    x axis line style= {-, line width = 0pt, color=black!0},
    y axis line style= {-, line width = 0pt, color=black!0},
    x tick style= {line width = \yafaxiswidth, color=yafaxiscolor, yshift = \yafaxispad},
    y tick style= {line width = \yafaxiswidth, color=yafaxiscolor, xshift = \yafaxispad},
    x tick label style = {font=\scriptsize, yshift = \yaftlpad},
    y tick label style = {font=\scriptsize, xshift = \yaftlpad},
    every axis y label/.style = {at = {(ticklabel cs:0.5)}, rotate=90, anchor=center, font=\scriptsize, yshift = -\yaflabelpad},
    every axis x label/.style = {at = {(ticklabel cs:0.5)}, anchor=center, font=\scriptsize, yshift = \yaflabelpad},
    x tick label style = {font=\scriptsize, yshift = 1pt},
    grid = major,
    major grid style  = {dash pattern = on 1pt off 3 pt},
	every axis plot post/.append style= {line width=\yafaxiswidth} ,
	legend cell align = left,
	legend style = {inner sep = 1pt, cells = {font=\scriptsize}},
	legend image code/.code={%
		\draw[mark repeat=2,mark phase=2,#1] 
		plot coordinates { (0cm,0cm) (0.15cm,0cm) (0.3cm,0cm) };%
	} 
}

\title{Mining Closed Strict Episodes\footnote{A preliminary version
appeared as ''Mining Closed Strict Episodes'', in Proceedings of
Tenth IEEE International Conference on Data Mining (ICDM 2010), 
2010~\cite{tatti:10:mining}.}}

\author{Nikolaj Tatti and Boris Cule}
\institute{Nikolaj Tatti \and Boris Cule \at University of Antwerp, Antwerp, Belgium,\\\email{nikolaj.tatti@ua.ac.be, boris.cule@ua.ac.be}}

\begin{document}
\maketitle

\begin{abstract}
Discovering patterns in a sequence is an important aspect of data mining. One
popular choice of such patterns are episodes, patterns in sequential data
describing events that often occur in the vicinity of each other. Episodes also
enforce in which order the events are allowed to occur.

In this work we introduce a technique for discovering closed episodes. Adopting
existing approaches for discovering traditional patterns, such as closed
itemsets, to episodes is not straightforward. First of all, we cannot define a
unique closure based on frequency because an episode may have several closed
superepisodes. Moreover, to define a closedness concept for episodes we need a
subset relationship between episodes, which is not trivial to define. 

We approach these problems by introducing strict episodes. We argue that
this class is general enough, and at the same time we are able to define a
natural subset relationship within it and use it efficiently. In order to mine closed
episodes we define an auxiliary closure operator. We show that this closure
satisfies the needed properties so that we can use the existing framework
for mining closed patterns. Discovering the true closed episodes can be done as
a post-processing step.  We combine these observations into an efficient mining
algorithm and demonstrate empirically its performance in practice.

\end{abstract}
\keywords{Frequent Episode Mining, Closed Episodes,  Level-wise Algorithm}

\section{Introduction}
\label{sec:introduction}

Discovering frequent patterns in an event sequence is an important field in
data mining. Episodes, as defined in~\cite{mannila:97:discovery}, represent a rich
class of sequential patterns, enabling us to discover events occurring in the vicinity
of each other while at the same time capturing complex interactions between the events.

More specifically, a frequent episode is traditionally considered to be a set
of events that reoccurs in the sequence within a window of a specified length.
Gaps are allowed between the events and the order in which the events are allowed
to occur is specified by the episode. Frequency, the number of windows in which
the episode occurs, is monotonically decreasing so we can use the well-known
level-wise approach to mine all frequent episodes.

The order restrictions of an episode are described by a directed acyclic graph
(DAG): the set of events in a sequence covers the episode if and only if each
event occurs only after all its parent events (with respect to the DAG) have
occurred (see the formal definition in Section~\ref{sec:prel}).
Usually, only two extreme cases are considered. A parallel episode poses no
restrictions on the order of events, and a window covers the
episode if the events occur in the window, in any order. In such a case,
the DAG associated with the episode contains no edges. The other extreme case is a
serial episode. Such an episode requires that the events occur in one, and only one,
specific order in the sequence. Clearly, serial episodes are more restrictive
than parallel episodes. If a serial episode is frequent, then its parallel version
is also frequent.

The advantage of episodes based on DAGs is that they allow us to capture dependencies
between the events while not being too restrictive.
\begin{example}
\label{ex:intro1}
As an example we will use text data, namely inaugural
speeches by presidents of the United States (see Section~\ref{sec:experiments}
for more details). Protocol requires the presidents to address the chief
justice and the vice presidents in their speeches. Hence, a pattern
\[
	\text{chief} \to \text{justic} \qquad \text{vice} \to \text{president}
\]
occurs in $10$ disjoint windows. This pattern captures the phrases 'chief justice' and 'vice president'
but because the address order varies from speech to speech, the pattern does not impose any
additional restrictions.
\end{example}

Episodes based on DAGs have, in practice, been over-shadowed by parallel and serial episodes, despite being defined at the same time~\cite{mannila:97:discovery}. The main
reason for this is the pattern explosion demonstrated in the
following example.
\begin{example}
To illustrate the pattern explosion we will again use inaugural
speeches by presidents of the United States.
By setting the window size to $15$ and the frequency threshold to $60$ we discovered 
a serial episode with $6$ symbols,
\[
	\text{preserv} \to \text{protect} \to \text{defend} \to \text{constitut} \to \text{unit} \to \text{state}.
\]
In total, we found another $4823$ subepisodes of size 6 of this episode.
However, all these episodes had only 3 distinct frequencies, indicating that
the frequencies of most of them could be derived from the frequencies of only 3
episodes, so the output could be reduced by leaving out $4821$ episodes. 

We illustrate the pattern explosion further in Table~\ref{fig:explosion}.
We see from the table that if the sequence has a frequent serial episode consisting of
$9$ labels, then mining frequent episodes will produce at least 100 million patterns.

\end{example}

\begin{table}[htb!]
\centering
\begin{tabular}{lrrrrrrrrr}
\toprule
pattern & 1 & 2 & 3 & 4 & 5 & 6 & 7 & 8 & 9\\
\midrule
itemsets & 1 & 3 & 7 & 15 & 31 & 63 & 127 & 255 & 511 \\
episodes & 1 & 4 & 16 & 84 & 652 & $7\,742$ & $139\,387$ & $3\,730\,216$ & $145\,605\,024$ \\
\bottomrule
\end{tabular}
\caption{Illustration of the pattern explosion. The first row is the number of
frequent itemsets produced by a single frequent itemset  with $n$ items. The second row is the number
of episodes produced by a single frequent serial episode with $n$ unique labels.
These numbers were obtained by a brute force enumeration.}
\label{fig:explosion}
\end{table}

Motivated by this example, we approach the problem of pattern explosion by
using a popular technique of closed patterns. A pattern is closed if there
exists no superpattern with the same frequency. Mining closed patterns has been
shown to reduce the output. Moreover, we can discover closed patterns efficiently.
However, adopting the concept of closedness to episodes is not without
problems.

\paragraph{Subset relationship}
Establishing a proper subset relationship is needed for two reasons. Firstly, to
make the mining procedure more efficient by discovering all possible subpatterns
before testing the actual episode, and secondly, to define a proper closure operator.

A na\"{i}ve approach to define whether an episode $G$ is a subepisode of an
episode $H$ is to compare their DAGs. This, however, leads to problems as the same
episode can be represented by multiple DAGs and a graph representing $G$ is not
necessarily a subgraph of a graph representing $H$ as demonstrated in the
following example.

\begin{example}
\label{ex:toy2}
Consider episodes $G_1$, $G_2$, and $G_3$ given in Figure~\ref{fig:toy}. Episode $G_1$ states that for a
pattern to occur $a$ must precede $b$ and $c$. Meanwhile, $G_2$ and $G_3$ state
that $a$ must be followed by $b$ and then by $c$.  Note that $G_2$ and $G_3$
represent essentially the same pattern that is more restricted than the
pattern represented by $G_1$. However, $G_1$ is a subgraph of $G_3$ but not
a subgraph of $G_2$. This reveals a problem if we base our definition of a subset
relationship of episodes solely on the edge subset relationship.
We solve this particular case by considering transitive closures, graphs in which each node must be connected to all its descendants by an edge,
thus ignoring graphs of form $G_2$. We will
not lose any generality since we are still going to discover episodes of form
$G_3$. Using transitive closure does not solve all problems for episodes containing
multiple nodes with the same label. For example, episodes $H_1$ and $H_2$ in Figure~\ref{fig:toy}
are the same even though their graphs are different.
\end{example}

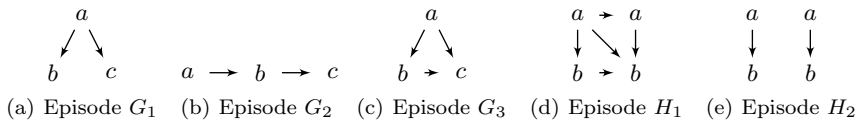
\begin{figure}[htb!]
\centering
\subfigure[Episode $G_1$\label{fig:toy:a}]{\begin{minipage}[b]{2.2cm}\centering\normalsize\begin{tikzpicture}[>=latex',line join=bevel,scale=0.3]\pgfsetlinewidth{0.5bp}\pgfsetcolor{black}
  \draw [->] (53.916bp,71.831bp) .. controls (49.939bp,63.877bp) and (45.185bp,54.369bp)  .. (36.207bp,36.413bp);
  \draw [->] (72.084bp,71.831bp) .. controls (76.061bp,63.877bp) and (80.815bp,54.369bp)  .. (89.793bp,36.413bp);
\begin{scope}
  \definecolor{strokecol}{rgb}{0.0,0.0,0.0};
  \pgfsetstrokecolor{strokecol}
  \draw (63bp,90bp) node {$a$};
\end{scope}
\begin{scope}
  \definecolor{strokecol}{rgb}{0.0,0.0,0.0};
  \pgfsetstrokecolor{strokecol}
  \draw (99bp,18bp) node {$c$};
\end{scope}
\begin{scope}
  \definecolor{strokecol}{rgb}{0.0,0.0,0.0};
  \pgfsetstrokecolor{strokecol}
  \draw (27bp,18bp) node {$b$};
\end{scope}\end{tikzpicture}\end{minipage}}
\subfigure[Episode $G_2$\label{fig:toy:b}]{\begin{minipage}[b]{2.2cm}\centering\normalsize\begin{tikzpicture}[>=latex',line join=bevel,scale=0.3]\pgfsetlinewidth{0.5bp}\pgfsetcolor{black}
  \draw [->] (54.003bp,18bp) .. controls (62.028bp,18bp) and (70.967bp,18bp)  .. (89.705bp,18bp);
  \draw [->] (144bp,18bp) .. controls (152.03bp,18bp) and (160.97bp,18bp)  .. (179.71bp,18bp);
\begin{scope}
  \definecolor{strokecol}{rgb}{0.0,0.0,0.0};
  \pgfsetstrokecolor{strokecol}
  \draw (27bp,18bp) node {$a$};
\end{scope}
\begin{scope}
  \definecolor{strokecol}{rgb}{0.0,0.0,0.0};
  \pgfsetstrokecolor{strokecol}
  \draw (207bp,18bp) node {$c$};
\end{scope}
\begin{scope}
  \definecolor{strokecol}{rgb}{0.0,0.0,0.0};
  \pgfsetstrokecolor{strokecol}
  \draw (117bp,18bp) node {$b$};
\end{scope}\end{tikzpicture}\end{minipage}}
\subfigure[Episode $G_3$\label{fig:toy:c}]{\begin{minipage}[b]{2.2cm}\centering\normalsize\begin{tikzpicture}[>=latex',line join=bevel,scale=0.3]\pgfsetlinewidth{0.5bp}\pgfsetcolor{black}
  \draw [->] (54bp,18bp) .. controls (56.615bp,18bp) and (59.229bp,18bp)  .. (71.93bp,18bp);
  \draw [->] (53.916bp,71.831bp) .. controls (49.939bp,63.877bp) and (45.185bp,54.369bp)  .. (36.207bp,36.413bp);
  \draw [->] (72.084bp,71.831bp) .. controls (76.061bp,63.877bp) and (80.815bp,54.369bp)  .. (89.793bp,36.413bp);
\begin{scope}
  \definecolor{strokecol}{rgb}{0.0,0.0,0.0};
  \pgfsetstrokecolor{strokecol}
  \draw (63bp,90bp) node {$a$};
\end{scope}
\begin{scope}
  \definecolor{strokecol}{rgb}{0.0,0.0,0.0};
  \pgfsetstrokecolor{strokecol}
  \draw (99bp,18bp) node {$c$};
\end{scope}
\begin{scope}
  \definecolor{strokecol}{rgb}{0.0,0.0,0.0};
  \pgfsetstrokecolor{strokecol}
  \draw (27bp,18bp) node {$b$};
\end{scope}\end{tikzpicture}\end{minipage}}
\subfigure[Episode $H_1$\label{fig:toy:d}]{\begin{minipage}[b]{2.2cm}\centering\normalsize\begin{tikzpicture}[>=latex',line join=bevel,scale=0.3]\pgfsetlinewidth{0.5bp}\pgfsetcolor{black}
  \draw [->] (27bp,71.831bp) .. controls (27bp,64.131bp) and (27bp,54.974bp)  .. (27bp,36.413bp);
  \draw [->] (54bp,18bp) .. controls (56.615bp,18bp) and (59.229bp,18bp)  .. (71.93bp,18bp);
  \draw [->] (45.169bp,71.831bp) .. controls (53.715bp,63.285bp) and (64.056bp,52.944bp)  .. (80.587bp,36.413bp);
  \draw [->] (99bp,71.831bp) .. controls (99bp,64.131bp) and (99bp,54.974bp)  .. (99bp,36.413bp);
  \draw [->] (54bp,90bp) .. controls (56.615bp,90bp) and (59.229bp,90bp)  .. (71.93bp,90bp);
\begin{scope}
  \definecolor{strokecol}{rgb}{0.0,0.0,0.0};
  \pgfsetstrokecolor{strokecol}
  \draw (27bp,90bp) node {$a$};
\end{scope}
\begin{scope}
  \definecolor{strokecol}{rgb}{0.0,0.0,0.0};
  \pgfsetstrokecolor{strokecol}
  \draw (99bp,90bp) node {$a$};
\end{scope}
\begin{scope}
  \definecolor{strokecol}{rgb}{0.0,0.0,0.0};
  \pgfsetstrokecolor{strokecol}
  \draw (27bp,18bp) node {$b$};
\end{scope}
\begin{scope}
  \definecolor{strokecol}{rgb}{0.0,0.0,0.0};
  \pgfsetstrokecolor{strokecol}
  \draw (99bp,18bp) node {$b$};
\end{scope}\end{tikzpicture}\end{minipage}}
\subfigure[Episode $H_2$\label{fig:toy:e}]{\begin{minipage}[b]{2.2cm}\centering\normalsize\begin{tikzpicture}[>=latex',line join=bevel,scale=0.3]\pgfsetlinewidth{0.5bp}\pgfsetcolor{black}
  \draw [->] (27bp,71.831bp) .. controls (27bp,64.131bp) and (27bp,54.974bp)  .. (27bp,36.413bp);
  \draw [->] (99bp,71.831bp) .. controls (99bp,64.131bp) and (99bp,54.974bp)  .. (99bp,36.413bp);
\begin{scope}
  \definecolor{strokecol}{rgb}{0.0,0.0,0.0};
  \pgfsetstrokecolor{strokecol}
  \draw (27bp,90bp) node {$a$};
\end{scope}
\begin{scope}
  \definecolor{strokecol}{rgb}{0.0,0.0,0.0};
  \pgfsetstrokecolor{strokecol}
  \draw (99bp,90bp) node {$a$};
\end{scope}
\begin{scope}
  \definecolor{strokecol}{rgb}{0.0,0.0,0.0};
  \pgfsetstrokecolor{strokecol}
  \draw (27bp,18bp) node {$b$};
\end{scope}
\begin{scope}
  \definecolor{strokecol}{rgb}{0.0,0.0,0.0};
  \pgfsetstrokecolor{strokecol}
  \draw (99bp,18bp) node {$b$};
\end{scope}\end{tikzpicture}\end{minipage}}
\caption{Toy episodes used in Example~\ref{ex:toy2}.}
\label{fig:toy}
\end{figure}

\paragraph{Frequency closure} Secondly, frequency does not satisfy the Galois
connection. In fact, given an episode $G$ there can be \emph{several} more
specific closed episodes that have the same frequency. So the closure
operator cannot be defined as a mapping from an episode to its frequency-closed version.
\begin{example}
\label{ex:toy3}
Consider sequence $s$ given in Figure~\ref{fig:tc:e} and episode $G_1$
given in Figure~\ref{fig:tc:a}.  Assume that we use a sliding window of size
$5$. There are two windows that cover episode $G_1$, namely $s[1, 5]$ and $s[6,
10]$, illustrated in Figure~\ref{fig:tc:e}. Hence, the frequency of $G_1$ is
$2$. There are \emph{two} serial episodes that are more specific than $G_1$ and
have the same frequency, namely, $G_2$ and $G_3$ given in Figures~\ref{fig:tc:b} and~\ref{fig:tc:c}.
Moreover, there is no superepisode of $G_2$ and $G_3$ that has frequency equal
to $2$.  In other words, we cannot define a unique closure for $G_1$ based on frequency.
\end{example}

\begin{figure}[htb!]
\centering
\subfigure[Episode $G_1$\label{fig:tc:a}]{\begin{minipage}[b]{3.4cm}\centering\normalsize\begin{tikzpicture}[>=latex',line join=bevel,scale=0.3]\pgfsetlinewidth{0.5bp}\pgfsetcolor{black}
  \draw [->] (54.003bp,53.101bp) .. controls (62.204bp,55.561bp) and (71.36bp,58.308bp)  .. (89.705bp,63.812bp);
  \draw [->] (144bp,26.101bp) .. controls (152.2bp,28.561bp) and (161.36bp,31.308bp)  .. (179.71bp,36.812bp);
  \draw [->] (54.003bp,36.899bp) .. controls (62.204bp,34.439bp) and (71.36bp,31.692bp)  .. (89.705bp,26.188bp);
  \draw [->] (144bp,63.899bp) .. controls (152.2bp,61.439bp) and (161.36bp,58.692bp)  .. (179.71bp,53.188bp);
\begin{scope}
  \definecolor{strokecol}{rgb}{0.0,0.0,0.0};
  \pgfsetstrokecolor{strokecol}
  \draw (27bp,45bp) node {$a$};
\end{scope}
\begin{scope}
  \definecolor{strokecol}{rgb}{0.0,0.0,0.0};
  \pgfsetstrokecolor{strokecol}
  \draw (117bp,18bp) node {$c$};
\end{scope}
\begin{scope}
  \definecolor{strokecol}{rgb}{0.0,0.0,0.0};
  \pgfsetstrokecolor{strokecol}
  \draw (117bp,72bp) node {$b$};
\end{scope}
\begin{scope}
  \definecolor{strokecol}{rgb}{0.0,0.0,0.0};
  \pgfsetstrokecolor{strokecol}
  \draw (207bp,45bp) node {$d$};
\end{scope}\end{tikzpicture}\end{minipage}}
\subfigure[Episode $G_2$\label{fig:tc:b}]{\begin{minipage}[b]{3.4cm}\centering\normalsize\begin{tikzpicture}[>=latex',line join=bevel,scale=0.3]\pgfsetlinewidth{0.5bp}\pgfsetcolor{black}
  \draw [->] (54.003bp,18bp) .. controls (62.028bp,18bp) and (70.967bp,18bp)  .. (89.705bp,18bp);
  \draw [->] (234bp,18bp) .. controls (242.03bp,18bp) and (250.97bp,18bp)  .. (269.71bp,18bp);
  \draw [->] (144bp,18bp) .. controls (152.03bp,18bp) and (160.97bp,18bp)  .. (179.71bp,18bp);
\begin{scope}
  \definecolor{strokecol}{rgb}{0.0,0.0,0.0};
  \pgfsetstrokecolor{strokecol}
  \draw (27bp,18bp) node {$a$};
\end{scope}
\begin{scope}
  \definecolor{strokecol}{rgb}{0.0,0.0,0.0};
  \pgfsetstrokecolor{strokecol}
  \draw (207bp,18bp) node {$c$};
\end{scope}
\begin{scope}
  \definecolor{strokecol}{rgb}{0.0,0.0,0.0};
  \pgfsetstrokecolor{strokecol}
  \draw (117bp,18bp) node {$b$};
\end{scope}
\begin{scope}
  \definecolor{strokecol}{rgb}{0.0,0.0,0.0};
  \pgfsetstrokecolor{strokecol}
  \draw (297bp,18bp) node {$d$};
\end{scope}\end{tikzpicture}\end{minipage}}
\subfigure[Episode $G_3$\label{fig:tc:c}]{\begin{minipage}[b]{3.4cm}\centering\normalsize\begin{tikzpicture}[>=latex',line join=bevel,scale=0.3]\pgfsetlinewidth{0.5bp}\pgfsetcolor{black}
  \draw [->] (144bp,18bp) .. controls (152.03bp,18bp) and (160.97bp,18bp)  .. (179.71bp,18bp);
  \draw [->] (54.003bp,18bp) .. controls (62.028bp,18bp) and (70.967bp,18bp)  .. (89.705bp,18bp);
  \draw [->] (234bp,18bp) .. controls (242.03bp,18bp) and (250.97bp,18bp)  .. (269.71bp,18bp);
\begin{scope}
  \definecolor{strokecol}{rgb}{0.0,0.0,0.0};
  \pgfsetstrokecolor{strokecol}
  \draw (27bp,18bp) node {$a$};
\end{scope}
\begin{scope}
  \definecolor{strokecol}{rgb}{0.0,0.0,0.0};
  \pgfsetstrokecolor{strokecol}
  \draw (117bp,18bp) node {$c$};
\end{scope}
\begin{scope}
  \definecolor{strokecol}{rgb}{0.0,0.0,0.0};
  \pgfsetstrokecolor{strokecol}
  \draw (207bp,18bp) node {$b$};
\end{scope}
\begin{scope}
  \definecolor{strokecol}{rgb}{0.0,0.0,0.0};
  \pgfsetstrokecolor{strokecol}
  \draw (297bp,18bp) node {$d$};
\end{scope}\end{tikzpicture}\end{minipage}}
\subfigure[Episode $G_4$\label{fig:tc:d}]{\begin{minipage}[b]{3.4cm}\centering\normalsize\begin{tikzpicture}[>=latex',line join=bevel,scale=0.3]\pgfsetlinewidth{0.5bp}\input{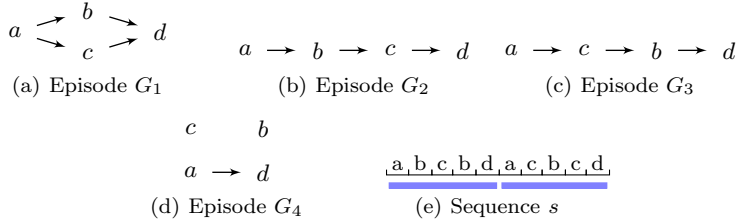}\end{tikzpicture}\end{minipage}}
\subfigure[\label{fig:tc:e}Sequence $s$]{
\centering
\begin{tikzpicture}
\node (n0) {};
\eventseq{a, b, c, b, d, a, c, b, c, d}{n}
\eventwin{n1}{n5}{3pt}{blue!50}{1pt}{1pt}
\eventwin{n6}{n10}{3pt}{blue!50}{1pt}{1pt}
\end{tikzpicture}}

\caption{Toy episodes used in Examples~\ref{ex:toy3}~and~\ref{ex:eclosure}. Edges induced by transitive closure are omitted to avoid clutter.}
\label{fig:tc}
\end{figure}

The contributions of our paper address these issues:
\begin{enumerate}
\item We introduce \emph{strict} episodes, a new subclass of general episodes.
We say that an episode is strict if all nodes with the same label are
connected. Thus all episodes in Figure~\ref{fig:toy} are strict, except $H_2$.
We will argue that this class is large, contains all serial and parallel episodes,
as well as episodes with unique labels, yet using only strict episodes eases the
computational burden.

\item We introduce a natural subset relationship between episodes based on the subset
relationship of sequences covering the episodes. We will prove that for strict
episodes this relationship corresponds to the subset relationship between
transitively closed graphs. For strict
episodes such a graph uniquely defines the episode.

\item We introduce milder versions of the closure concept, including the
\emph{instance-closure}. We will show that these closures can be used
efficiently, and that a frequency-closed episode is always instance-closed\footnote{In~\cite{tatti:10:mining}, the closure was based only on adding edges whereas in this version
we are also adding nodes.}.
We demonstrate that computing closure and frequency can be done in polynomial time\footnote{This was not guaranteed in\cite{tatti:10:mining}.}.

\item Finally, we present an algorithm that generates strict
instance-closed episodes with transitively closed graphs. Once these episodes
are discovered we can further prune the output by removing the episodes that
are not frequency-closed.
\end{enumerate}

\section{Preliminaries and Notation}
\label{sec:prel}
We begin by presenting the preliminary concepts and notations that will be used
throughout the paper. In this section we introduce the notions
of sequence and episodes.

A \emph{sequence} $s = s_1\cdots s_L$ is a string of symbols, or \emph{events}, coming from an 
\emph{alphabet} $\Sigma$, so that for each $i$, $s_i \in \Sigma$. Given a strictly increasing mapping
$\funcdef{m}{[1, N]}{[1, L]}$ we will define $s_m$ to be a subsequence $s_{m(1)}\cdots s_{m(N)}$.
Similarly, given two integers $1 \leq a \leq b \leq L$ we define $s[a, b] = s_a\cdots s_b$. 

An \emph{episode} $G$ is represented by a directed acyclic graph with labelled nodes,
that is, $G = (V, E, \lab{})$, where $V = \enpr{v_1}{v_K}$ is the set of nodes,
$E$ is the set of directed edges, and $\lab{}$ is the function
$\funcdef{\lab{}}{V}{\Sigma}$, mapping each node $v_i$ to its label. We denote
the set of nodes of an episode $G$ with $V(G)$, and its set of edges with $E(G)$.

Given a sequence $s$ and an episode $G$ we say that $s$ \emph{covers} $G$, or $G$ \emph{occurs} in $s$, if 
there is an \emph{injective} map $f$ mapping each node $v_i$ to a valid index
such that the node $v_i$ in $G$ and the corresponding sequence element
$s_{f(v_i)}$ have the same label, $s_{f(v_i)} = \lab{v_i}$, and that if there
is an edge $(v_i, v_j)$ in $G$, then we must have $f(v_i) < f(v_j)$. In other
words, the parents of $v_j$ must occur in $s$ before $v_j$.  For an example, see
Figure~\ref{fig:def:a}. If the mapping $f$ is surjective, that is, all events
in $s$ are used, we will say that $s$ is an \emph{instance} of $G$.

\subsection{Frequency}
In our search for frequent episodes, we will use and compare two conceptually different definitions of frequency.

Traditionally, episode mining is based on searching for episodes that are covered
by windows of certain fixed size often enough. The frequency of a
given episode is then defined as the number of such windows that cover it.

\begin{definition}
The \emph{fixed-window frequency} of an episode $G$ in a sequence $s$, denoted $\freqf{G}$, is defined as the number of time windows of a given size $\rho$ within $s$, in which the episode $G$ occurs. Formally,
\[
	\freqf{G} = \abs{\set{(a, b) | b = a + \rho - 1, a \leq N, b \geq 1, \text{ and } s[a, b] \text{ covers } G }}.
\]
\end{definition}
See Figure~\ref{fig:def:b} for example.

The frequency of an episode is sometimes expressed using the number of minimal windows that contain it. To satisfy the downward-closed property, we say that these windows must be non-overlapping.

\begin{definition}
The \emph{disjoint-window frequency} of an episode $G$ in a sequence $s$,
denoted $\freqm{G}$, is defined as the maximal number of non-overlapping windows within $s$ that contain episode $G$. Formally,
\[
	\freqm{G} = \max \Set{\abs{\enset{(a_1, b_1)}{(a_N, b_N)}} |\begin{array}{l}  s[a_i, b_i] \text{ covers } G, b_i - a_i < \rho \text{ and} \\
	{[a_i, b_i]} \cap [a_j, b_j] = \emptyset \text{ for } 1 \leq i, j \leq N \end{array}}.
\]
\end{definition}
See Figure~\ref{fig:def:c} for example.

We now establish a connection between the disjoint-window frequency and actual minimal windows.
\begin{definition}
Given a sequence $s$ and an episode $G$, a window $s[a,b]$ is called a \emph{minimal window} of $G$ in $s$, if $s[a,b]$ covers $G$, and if no proper subwindow of $s[a,b]$ covers $G$.
We will also refer to the interval $[a, b]$ as a minimal window, if $s$ is clear from the context.
\end{definition}
It is easy to see that the maximal number of non-overlapping windows within $s$ that contain $G$ is equal to the maximal number of non-overlapping minimal windows within $s$ that contain $G$.

Whenever it does not matter whether we are dealing with the fixed-window frequency or the disjoint-window frequency, we simply denote $\freq{G}$.

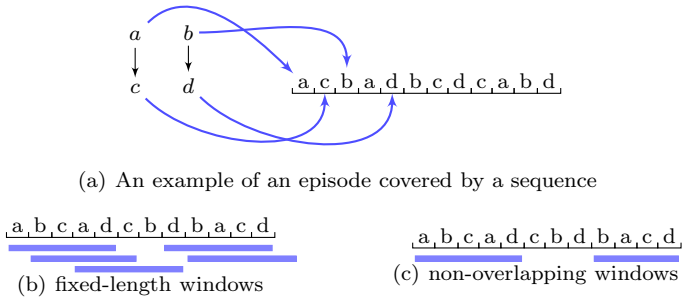
\begin{figure}[htb!]
\centering
\subfigure[\label{fig:def:a}An example of an episode covered by a sequence]{
\begin{minipage}{8cm}
\centering
\begin{tikzpicture}

\node[anchor = mid] (e1) {$a$};
\node[right=0.7cm of e1.mid, inner sep = 2pt, anchor = mid] (e2) {$b$};
\node[below=0.7cm of e1.mid, inner sep = 2pt, anchor = mid] (e3) {$c$};
\node[below=0.7cm of e2.mid, inner sep = 2pt, anchor = mid] (e4) {$d$};
\path [draw, -latex'] (e1) -- (e3);
\path [draw, -latex'] (e2) -- (e4);

\node[right=1cm of e4] (n0) {};
\eventseq{a, c, b, a, d, b, c, d, c, a, b, d}{n}

\path [line] (e1) edge [out=45,in=135] (n1);
\path [line] (e2) edge [out=0,in=90] (n3);
\path [line] (e3) edge [out=315,in=270] (n2);
\path [line] (e4) edge [out=315,in=270] (n5);
\end{tikzpicture}
\end{minipage}}

\subfigure[\label{fig:def:b}fixed-length windows]{
\begin{minipage}{5cm}
\centering
\begin{tikzpicture}
\node (n0) {};
\eventseq{a, b, c, a, d, c, b, d, b, a, c, d}{n}
\eventwin{n1}{n5}{3pt}{blue!50}{1pt}{1pt}
\eventwin{n2}{n6}{7pt}{blue!50}{1pt}{1pt}
\eventwin{n4}{n8}{11pt}{blue!50}{1pt}{1pt}

\eventwin{n8}{n12}{3pt}{blue!50}{1pt}{1pt}
\eventwin{n9}{n12}{7pt}{blue!50}{1pt}{-8pt}
\end{tikzpicture}
\end{minipage}}
\subfigure[\label{fig:def:c}non-overlapping windows]{
\begin{minipage}{5cm}
\centering
\begin{tikzpicture}
\node (n0) {};
\eventseq{a, b, c, a, d, c, b, d, b, a, c, d}{n}
\eventwin{n1}{n5}{3pt}{blue!50}{1pt}{1pt}
\eventwin{n9}{n12}{3pt}{blue!50}{1pt}{1pt}
\end{tikzpicture}
\end{minipage}}

\caption{A toy example illustrating different support measures.  
Figure~\ref{fig:def:a} contains an example of a sequence covering an episode.
Figure~\ref{fig:def:b} shows all 5 sliding windows of length 5 containing the
episode. Figure~\ref{fig:def:c} shows the maximal number, 2, of non-overlapping
windows covering the episode.}

\end{figure}

\section{Strict Episodes}
\label{sec:strict}
In this section we will define our mining problem and give a rough outline of
the discovery algorithm.

Generally, a pattern is considered closed if there exists no more specific
pattern having the same frequency. In order to speak of more specific patterns,
we must first have a way to describe episodes in these terms.

\begin{definition}
Assume two transitively closed episodes $G$ and $H$ with the same number of
nodes.  An episode $G$ is called a \emph{subepisode} of episode $H$, denoted
$G\preceq H$ if the set of all sequences that cover $H$ is a subset of the set
of all sequences that cover $G$. If the set of all sequences that cover $H$ is a proper subset of the set   
of all sequences that cover $G$, we call $G$ a \emph{proper subepisode} of $H$, denoted
$G\prec H$.

For a more general case, assume that $\abs{V(G)} < \abs{V(H)}$.  We
say that $G$ is a subepisode of $H$, denoted $G \preceq H$, if there is a subgraph $H'$ of $H$
such that $G \preceq H'$. Moreover, let $\alpha$ be a graph homomorphism from $H'$ to $H$.
If we wish to emphasize $\alpha$, we write $G \preceq_\alpha H$. 

If  $\abs{V(G)} > \abs{V(H)}$, then $G$ is automatically not a subepisode of $H$.
\end{definition}

The problem with this definition is that we do not have the means to
compute this relationship for general episodes. To do this, one would have
to enumerate all possible sequences that cover $H$ and compute whether they cover $G$.
We approach this problem by restricting ourselves to a class of episodes where
this comparison can be performed efficiently.

\begin{definition}
An episode $G$ is called \emph{strict} if for any two nodes $v$ and
$w$ in $G$ sharing the same label, there exists a path either from $v$ to $w$
or from $w$ to $v$.
\end{definition}

We will show later that the subset relationship can be computed efficiently for
strict episodes.  However, as can be seen in
Figure~\ref{fig:toy2}, there are episodes that are not strict.  Our algorithm
will not discover these types of patterns.

\begin{figure}[htb!]
\centering
\subfigure[non-strict]{\begin{minipage}[b]{2.5cm}\centering\normalsize\begin{tikzpicture}[>=latex',line join=bevel,scale=0.3]\pgfsetlinewidth{0.5bp}\pgfsetcolor{black}
  \draw [->] (27bp,71.831bp) .. controls (27bp,64.131bp) and (27bp,54.974bp)  .. (27bp,36.413bp);
  \draw [->] (99bp,71.831bp) .. controls (99bp,64.131bp) and (99bp,54.974bp)  .. (99bp,36.413bp);
\begin{scope}
  \definecolor{strokecol}{rgb}{0.0,0.0,0.0};
  \pgfsetstrokecolor{strokecol}
  \draw (27bp,90bp) node {$a$};
\end{scope}
\begin{scope}
  \definecolor{strokecol}{rgb}{0.0,0.0,0.0};
  \pgfsetstrokecolor{strokecol}
  \draw (99bp,18bp) node {$c$};
\end{scope}
\begin{scope}
  \definecolor{strokecol}{rgb}{0.0,0.0,0.0};
  \pgfsetstrokecolor{strokecol}
  \draw (99bp,90bp) node {$a$};
\end{scope}
\begin{scope}
  \definecolor{strokecol}{rgb}{0.0,0.0,0.0};
  \pgfsetstrokecolor{strokecol}
  \draw (27bp,18bp) node {$b$};
\end{scope}\end{tikzpicture}\end{minipage}}
\subfigure[strict]{\begin{minipage}[b]{2.5cm}\centering\normalsize\begin{tikzpicture}[>=latex',line join=bevel,scale=0.3]\pgfsetlinewidth{0.5bp}\pgfsetcolor{black}
  \draw [->] (27bp,71.831bp) .. controls (27bp,64.131bp) and (27bp,54.974bp)  .. (27bp,36.413bp);
  \draw [->] (54bp,90bp) .. controls (56.615bp,90bp) and (59.229bp,90bp)  .. (71.93bp,90bp);
  \draw [->] (99bp,71.831bp) .. controls (99bp,64.131bp) and (99bp,54.974bp)  .. (99bp,36.413bp);
  \draw [->] (45.169bp,71.831bp) .. controls (53.715bp,63.285bp) and (64.056bp,52.944bp)  .. (80.587bp,36.413bp);
\begin{scope}
  \definecolor{strokecol}{rgb}{0.0,0.0,0.0};
  \pgfsetstrokecolor{strokecol}
  \draw (27bp,90bp) node {$a$};
\end{scope}
\begin{scope}
  \definecolor{strokecol}{rgb}{0.0,0.0,0.0};
  \pgfsetstrokecolor{strokecol}
  \draw (99bp,18bp) node {$c$};
\end{scope}
\begin{scope}
  \definecolor{strokecol}{rgb}{0.0,0.0,0.0};
  \pgfsetstrokecolor{strokecol}
  \draw (99bp,90bp) node {$a$};
\end{scope}
\begin{scope}
  \definecolor{strokecol}{rgb}{0.0,0.0,0.0};
  \pgfsetstrokecolor{strokecol}
  \draw (27bp,18bp) node {$b$};
\end{scope}\end{tikzpicture}\end{minipage}}
\caption{An example of a non-strict and a strict episode.}
\label{fig:toy2}
\end{figure}
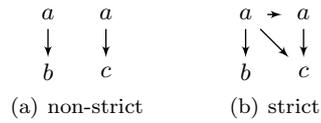

Having defined a subset relationship, we can now define an $f$-closed episode.

\begin{definition}
An episode $G$ is \emph{frequency-closed}, or \emph{f-closed}, if there exists no episode $H$, such that $G\prec H$ and $\freq{G} = \freq{H}$.
\end{definition}

\begin{problem}
Given a sequence $s$, a frequency measure, either fixed-window or
disjoint-window, and a threshold $\sigma$, find all $f$-closed strict episodes
from $s$ having the frequency higher or equal than $\sigma$.
\end{problem}

A traditional approach to discovering closed patterns is to discover generators,
that is, for each closed pattern $P$, discover minimal patterns whose closure
is equal to $P$~\cite{pasquier:99:discovering}. When patterns are itemsets, it
holds that the collection of frequent generators are downward closed. Hence,
they can be mined efficiently using a BFS-style approach. 

We cannot directly apply this framework for two reasons:
Firstly, unlike with itemsets, we cannot define a closure based on frequency.
We solve this by defining an instance-closure, a more conservative closure that
guarantees that all $f$-closed episodes are discovered. Once instance-closed
episodes are discovered, the $f$-closed episodes are selected in a post-processing step.
The second obstacle is the fact that the collection of generator episodes is not
necessarily downward-closed. We solve this problem by additionally using some intermediate
episodes that will guarantee the correctness of the algorithm.

A sketch of the miner is given in Algorithm~\ref{alg:sketch} (the details
of the algorithm are described in subsequent sections). 
The algorithm consists of two loops. In the outer loop,
Lines~\ref{alg:sk:loop1begin}--\ref{alg:sk:loop1end}, we discover parallel
episodes by adding nodes. In the inner loop,
Lines~\ref{alg:sk:loop2begin}--\ref{alg:sk:loop2end}, we discover general
episodes by adding edges. Each candidate episode is tested, and if the
candidate is frequent and a generator, then the episodes is added into the
collection. Finally, we discover the $f$-closed episodes as a last step.

\begin{algorithm}
\Input{sequence $s$, threshold $\sigma$, window size $\rho$}
\Out{frequent $f$-closed episodes}
	$\efam{C} \define $ all frequent episodes with $1$ node\;
	$\efam{E} \define \emptyset$; $N \define 1$\;
	\While {$\efam{C} \neq \emptyset$} {\nllabel{alg:sk:loop1begin}
		$M \define 0$\;
		\While {$\efam{C} \neq \emptyset$} { \nllabel{alg:sk:loop2begin}
			\ForEach {$G \in \efam{C}$} {
				\If {$G$ is a frequent generator w.r.t. $i$-closure} {
					add $G$ to $\efam{E}$\;
					add intermediate episodes\;
				}
			}
			$\efam{P} \define $ episodes with $N$ nodes and $M$ edges from $\efam{E}$\;
			$\efam{C} \define $ candidates generated from $\efam{P}$ with $N$ nodes and $M + 1$ edges\;\nllabel{alg:sk:generate}
			$M \define M + 1$\;\nllabel{alg:sk:loop2end}
		}
		$\efam{P} \define $ parallel episodes with $N$ nodes from $\efam{E}$\;
		$\efam{C} \define $ parallel candidates generated from $\efam{P}$ with $N + 1$ nodes\;
		$N \define N + 1$\; \nllabel{alg:sk:loop1end}
	}
	\Return $f$-closed episodes from the $i$-closures of episodes in $\efam{E}$\;
\caption{Rough outline of the breath-first mining algorithm. The details of each step are given in Sections~\ref{sec:subset}--\ref{sec:algorithm}.}
\label{alg:sketch}
\end{algorithm}

To complete the algorithm we need to solve several subproblems:
\begin{enumerate}
\item computing the subset relationship efficiently (Section~\ref{sec:subset})
\item defining and computing instance-closure (Section~\ref{sec:closure})
\item generating candidate episodes,  Line~\ref{alg:sk:generate} (Section~\ref{sec:candidate})
\item generating intermediate episodes and proving the correctness (Section~\ref{sec:algorithm})
\end{enumerate}

\section{Computing The Subset Relationship}
\label{sec:subset}

In this section we will demonstrate that computing the subset relationship for strict episodes
can be done efficiently. This allows us to build an algorithm to efficiently
discover closed episodes.

We now provide a canonical form for episodes, which will help us in further theorems and algorithms.
We define an episode that has the maximal number of edges using a
fundamental notion familiar from graph theory.

\begin{definition}
The \emph{transitive closure}
of an episode $G = (V, E, \lab{})$ is an episode
$\closure{G}$, where $G$ and $\closure{G}$ have the same set of nodes $V$, the
same $\lab{}$ function mapping nodes to labels, and the set of edges in
$\closure{G}$ is equal to
\[
E(\closure{G})   = E\cup \set{(v_i, v_j) \mid \text{ a path exists in } G \text{ from } v_i \text{ to } v_j}.
\]
\end{definition}

Note that, despite its name, the transitive closure has nothing to do with the concept of closed episodes.

\begin{definition}
Let $\sspace$ be the space of all strict and transitively closed episodes. 
\end{definition}
In the remaining text, we consider episodes to be transitively closed and 
strict, unless stated otherwise.  An episode and its transitive closure will
always have the same frequency, hence by restricting ourselves to tranistively
closed episodes we will not lose any episodes.

For notational simplicity, we now introduce the concept of two episodes having
\emph{identical nodes}. Given an episode $G \in \sspace$ with nodes $V(G) =
\enset{v_1}{v_N}$, we assume from now on that the order of the nodes is always
fixed such that for $i < j$ either $\lab{v_i} < \lab{v_j}$ lexicographically, or $\lab{v_i} =
\lab{v_j}$ and $v_i$ is an ancestor of $v_j$ with respect to $E(G)$ (i.e. edge $(v_i,v_j)\in E(G)$).  We say that two episodes $G$ and   
$H$, with $V(G) = \enset{v_1}{v_N}$ and $V(H) = \enset{w_1}{w_N}$ have
\emph{identical nodes} if $\lab{v_i} = \lab{w_i}$ for $i = \enset{1}{N}$. To simplify notation, we
often identify $v_i$ and $w_i$. This convention allows us to write statements
such as $E(G) \cup E(H)$, if $G$ and $H$ have identical nodes.

Our next step is to show how we can test the subset relationship for strict episodes.

\begin{lemma}
\label{lem:mappings}
Let $G, H \in \sspace$ be episodes with identical nodes. Let $s$
be a valid instance of both $G$ and $H$. Let $g$ and $h$ be the corresponding
functions mapping nodes of $G$ and $H$ to indices of $s$, respectively. Then $g = h$.
\end{lemma}

\begin{proof}
Let $v$ be a node. Assume that function $g$ maps $v$ to the $l$th occurence of $\lab{v}$ in
$s$. Since $s$ is an instance, then there are $l - 1$ ancestors of $v$ in $G$
having the same label as $v$. Since $G$ and $H$ have identical nodes, $v$ also
has $l - 1$ ancestors in $H$. Since $s$ is an instance of $H$, $h$ must map
$v$ to $l$th occurrence of $\lab{v}$.  This implies that $g = h$.
\end{proof}

Lemma~\ref{lem:mappings} implies that given an episode $G$ and an instance $s$,
there is only one valid function $f$ mapping nodes of $G$ to indices of $s$.  Let us
denote this mapping by $\map{G, s} = f$. If $G$ is a parallel episode with
nodes $V = V(G)$ we write $\map{V, s}$.

Crucially, we can easily compute the subset relationship between two episodes.

\begin{theorem}
\label{thr:sub}
For episodes $G, H \in \sspace$ with identical nodes, $E(G)\subseteq E(H)$ if and only if $G \preceq H$. 
\end{theorem}

\begin{proof}
To prove the ''only if'' direction assume that $E(G) \subseteq E(H)$. 
Let $s = \enset{s_1}{s_N}$ be an instance of $H$ and let $f = \map{H, s}$ be the corresponding mapping.
Then $f$ is also a valid mapping for $G$. Thus, $G \preceq H$. 

To prove the other direction, assume that $E(G) \nsubseteq E(H)$. We therefore
must have an edge $e = (x, y) \in E(G)$, such that $e \notin E(H)$.  We build
$s$ by first visiting every parent of $y$ in $H$ in a valid order with respect
to $H$, then $y$ itself, and then the rest of the nodes, also in a valid order.
Let $h$ be the visiting order of $G$ while constructing $s$, that is, $h(v) = 1$,
if we visited $v$ first, $h(v) = 2$, if we visited $v$ second.
Note that $h(y) < h(x)$.
Assume now that $s$ covers $G$ and let $f = \map{G, s}$ be the corresponding mapping. 
But then Lemma~\ref{lem:mappings} implies that $g = h$, thus $g(y) < g(x)$,
contradicting the fact that $(x, y) \in E(G)$.
\end{proof}

Theorem~\ref{thr:sub} essentially shows that our subset relationship is in fact
a graph subset relationship which allows us to design an efficient mining
algorithm.

We finish this section by defining what we exactly mean when we say that two
episodes are equivalent and demonstrate that the class of strict episodes
contains all parallel episodes. 

\begin{definition}
Episodes $G$ and $H$ are said to be equivalent, denoted by $G \sim H$, if each sequence that covers $G$ also covers $H$, and vice versa.
\end{definition}

\begin{corollary}[of Theorem~\ref{thr:sub}]
For episodes $G, H \in \sspace$, $G \sim H$ if and only if $E(G) = E(H)$ and $G$ and $H$ have identical nodes. 
\end{corollary}

\begin{proof}
This follows from the fact that $G \sim H$ is equivalent to $G \preceq H$ and
$H \preceq G$, and that $E(G) = E(H)$ is equivalent to $E(G) \subseteq E(H)$ and $E(H)
\subseteq E(G)$.
\end{proof}

Note that by generating only transitively closed strict episodes, we
have obtained an efficient way of computing the subset relationship between two episodes.
At first glance,
though, it may seem that we have completely omitted certain parallel episodes
from consideration --- namely, all non-strict parallel episodes (i.e.
those containing multiple nodes with same labels). Note, however, that for each
such episode $G$, there exists a strict episode $H$, such that $G \sim
H$. To build such an episode $H$, we just need to create edges that would
strictly define the order among nodes with the same labels. From now on, when we talk of parallel episodes, we actually refer to their strict equivalents.

\section{Closure}
\label{sec:closure}

Having defined a subset relationship among episodes, we are now able to speak
of an episode being more specific than another episode. However, this is only
the first step towards defining the closure of an episode. We know that the
closure must be more specific, but it must also be unique and well-defined. We
have already seen that basing such a closure on the frequency fails, as there
can be multiple more specific closed episodes that could be considered as
closures.

In this section we will establish three closure operators, based on the instances of the episode found within the sequence. 
The first closure adds
nodes, the second one adds edges, and the third is a combination of the first two. We also show that these operators satisfy three important properties.
We will use these properties to prove the correctness of our mining algorithm.
The needed properties for a closure operator $h$ are
\begin{enumerate}
\item \emph{Extension:} $G \preceq h(G)$,
\item \emph{Idempotency:} $h(G) = h(h(G))$,
\item \emph{Monotonicity:} $G_1 \preceq G_2 \Rightarrow h(G_1) \preceq h(G_2)$.
\end{enumerate}

These properties are usually shown using the Galois connection but to avoid cumbersome
notation we will prove them directly.

\subsection{Node Closure}

In this section we will define a node closure and show that it satisfies the
properties. Assume that we are given a sequence $s$ and a window size $\rho$.

Our first step is to define a function $f_N$ which maps an episode $G \in \sspace$ to a set
of intervals which contain all instances of $G$,
\[
	f_N(G; s) = \set{[\min m, \max m] \mid s_m \text{ covers } G, \ \max m - \min m < \rho, m \in M}, 
\]
where $M$ contains all strictly increasing mappings to $s$. 

Our next step is to define $X_G$ to be the set of all symbols occurring in each interval,
\[
	X_G = \set{x \in \Sigma \mid x \text{ occurs in } s[a, b] \text{ for all } [a, b] \in f_N(G)}. 
\]

Let $W$ be the labels of the nodes of $G$. We define our first closure operator, 
$\iclN{G}$ to be $G$ augmented with nodes having the labels $X_G -
W$, that is, we add nodes to $G$ with labels that occur inside each
window that contains $G$.

\begin{theorem}
$\iclN{G}$ is an idempotentic and monotonic extension operator.
\end{theorem}

\begin{proof}
The extension property follows immediately because we are only adding new nodes.

Assume now that $G \preceq H$. Let $[a, b] \in f_N(H)$ be an interval. Then,
there is an interval $[c, d] \in f_N(G)$ such that $a \leq c \leq d \leq b$.
This means that any symbol occurring in every interval in $f_N(G)$ 
will also occur in every interval in $f_N(H)$, that is, $X_G \subseteq X_H$.

Let $\alpha$ be a graph homomorphism such that $G \preceq_\alpha H$.  Let $x
\in X_G$ be a symbol not occurring in $G$ and let $v$ be the new node in
$\iclN{G}$ with this label.  If $x$ does not occur in $H$, then $x \in X_H$ and thus a node with a label $x$ is added into $H$. In any case, there
is a node $w$ with a label $x$ in $\iclN{H}$.  We can extend $\alpha$ by
setting $\alpha(v) = w$.  By doing this for each new node we have proved monotonicity, i.e. that
$\iclN{G} \preceq \iclN{H}$. 

To prove idempotency, let us write $H = \iclN{G}$. Since any new node in $H$
must occur inside the instances of $G$, we have $f_N(G) \subseteq f_N(H)$.  This
implies that $X_H \subseteq X_G$ and since we saw before that $X_G \subseteq
X_H$, it implies that $X_G = X_H$.  Since for every label in $X_H$ there is a
node in $H$ with the same label, it holds that $\iclN{H} = H$.
\end{proof}

Note that the node closure adds only events with unique labels. The reason for
this is that if we add node $x$ to an episode containing node $y$ such that
$\lab{x} = \lab{y}$, then we would have to connect $x$ and $y$. This may reduce
the instances and invalidate the proof. For the same reason, we only add a maximum of one new node with a particular label. In other words, if each window containing episode $G$ also contains two occurrences of $a$, we will only add one node with label $a$ to $G$ (provided $G$ does not contain a node labelled $a$ already).

\subsection{Edge Closure}

We begin by introducing the concept of a maximal episode that is covered by a given set of sequences. 

\begin{definition}
\label{def:maximalepisode} 
Given a set of nodes $V$, and a set S of instances of V, interpreted as a parallel episode, we define the \emph{maximal episode}
covered by set $S$ as the episode $H$, where $V(H) = V$ and
\[
	E(H) = \set{(x, y) \in V \times V \mid f(x) < f(y), f = \map{V, s} \text{ for all } s \in S}, 
\]
where $\map{V, s}$ refers to the mapping defined in Lemma~\ref{lem:mappings}
and $V$ is interpreted as a parallel episode. 
\end{definition}

To define a closure operator we first define a function mapping an episode $G$ to all of its valid instances in a sequence $s$,
\[
	f_E(G; s) = \set{s_m \mid s_m \text{ covers } G, \ \max m - \min m < \rho, m \in M},
\]
where $M$ contains all strictly increasing mappings to $s$. When the sequence is known from the context, we denote simply $f_E(G)$

We define $\iclE{G}$ to be the maximal episode covered by $f_E(G)$.
If $\iclE{G} = G$, then we call $G$ an $e$-closed episode.

\begin{theorem}
$\iclE{G}$ is an idempotentic and monotonic extension operator.
\end{theorem}

\begin{proof}
To prove the extension property assume an edge $(v_i, v_j) \in E(G)$.
Let $V$ be the nodes in $G$.
Let $w \in f_E(G)$ be an instance of $G$ and let $f = \map{V, w}$ be the corresponding mapping.
Lemma~\ref{lem:mappings} implies that $\map{V, w} = \map{G, w}$. Hence $f(v_i) < f(v_j)$. 
Since this holds for every map, we have $(v_i, v_j) \in E(\iclE{G})$.

To prove the idempotency, let $H = \iclE{G}$. The extension property implies
that $G \preceq H$ so by definition $f_E(H) \subseteq f_E(G)$. But any instance
in $f_E(G)$ also covers $H$. Thus, $f_E(H) \subseteq f_E(G)$ and so
$f_E(H) = f_E(G)$. This implies the idempotency.

Assume now that $G \preceq H$.  Let $\alpha$ be the graph homomorphism such
that $G \preceq_\alpha H$. We will show that $\iclE{G} \preceq_\alpha
\iclE{H}$.  Let $(x, y) \in E(\iclE{G})$. Let $w$ be an instance of $H$ and let
$f = \map{H, s}$ the corresponding mapping to $w$.
Assume that $f(\alpha(x)) \geq f(\alpha(y))$. Let $v$ be the subsequence of $w$
containing only the indices in the range of $f \circ \alpha$. Note that $v$ is a valid   
instance of $G$ and $f \circ \alpha = \map{V(G), v}$.
This contradicts the fact that $(x, y) \in E(\iclE{G})$. Hence, $f(\alpha(x)) < f(\alpha(y))$.
This implies that $(\alpha(x), \alpha(y)) \in E(\iclE{H})$ which completes the proof.
\end{proof}

\begin{example}
\label{ex:eclosure}
Consider sequence $s$ given in Figure~\ref{fig:tc:e} and episode $G_4$
given in Figure~\ref{fig:tc:d} and assume that the window length is $5$.  There
are four instances of $G_4$ in $s$, namely $abcd$, $acdb$, $acbd$ and $abcd$. 
Therefore, $f_E(G_4) = \set{abcd,acbd}$. The serial episodes corresponding
to these subsequences are $G_2$ and $G_3$ given in Figure~\ref{fig:tc}.
By taking the intersection of these two episodes we obtain $G_1 = \iclE{G_4}$
given in Figure~\ref{fig:tc:a}.
\end{example}

\subsection{Combining Closures}

We can combine the node closure and the edge closure into one operator.

\begin{definition}
Given an episode $G \in \sspace$, we define $\iclEN{G} = \iclE{\iclN{G}}$. To simplify the
notation, we will refer to 
$\iclEN{G}$ as $\iclosure{G}$, the \emph{$i$-closure} of $G$.  We will say that
$G$ is \emph{$i$-closed} if $G = \iclosure{G}$.
\end{definition}

\begin{theorem}
$\iclosure{G}$ is an idempotentic and monotonic extension operator.
\end{theorem}

\begin{proof}
The extension and monotonicity properties follow directly from the fact that both
$\iclN{G}$ and $\iclE{G}$ are monotonic extension operators.

To prove idempotency let $H = \iclosure{G}$ and $H' = \iclN{G}$. Since any
instance of $H = \iclE{H'}$ is also an instance of $H'$, and, per definition, vice versa, we see 
that $f_N(H) = f_N(H')$, and consequently $\iclN{H} = H$, and 
$\iclosure{H} = \iclE{\iclN{H}} = \iclE{H} = H$.
\end{proof}

The advantage of mining $i$-closed episodes instead of $e$-closed is prominent
if the sequence contains a long sequential pattern. More specifically, assume that
the input sequence contains a frequent subsequence of $N$ symbols $p_1, \ldots, p_N$, and no other
permutation of this pattern occurs in the sequence. The number of $e$-closed subpatterns of this subsequence 
is $2^N - 1$, namely, all non-empty serial subepisodes. However, the number of
its $i$-closed subpatterns is $N(N + 1)/2$, namely, serial episodes of form $p_i \to \cdots \to p_j$ for $1 \leq i \leq j \leq N$.

\subsection{Computing Closures}
During the mining process, we need to compute the closure of an episode. The
definition of closures use $f_N(G)$ and $f_E(G)$ which are based on instances
of $G$ in $s$.  However, there can be an exponential number of such instances in $s$.

In the following discussion we will often use the following notations.
Given an episode $G$, we write $G + v$ to mean the
episode $G$ augmented with an additional node $v$.  Similarly, we will use the
notations $G + e$ and $G + V$, where $e$ is an edge and $V$ is a set of nodes.
We also use $G - v$, $G - V$, and $G - e$ to mean episodes where either nodes
or edges are removed.

To avoid the problem of an exponential number of instances we make two
observations: Firstly, a node with a new label $l$ is added into the closure if
and only if it occurs in every \emph{minimal window} of $G$. Secondly, an edge
$(x, y)$ is added into the closure if and only if there is no minimal window
for $G + (y, x)$. Thus to compute the closure we need an algorithm that finds
all minimal windows of $G$ in $s$. Note that, unlike with instances of $G$,
there can be only $\abs{s}$ minimal windows.

Using strict episodes allows us to discover minimal windows in an efficient greedy
fashion. 

\begin{lemma}
\label{lem:greedy}
Given an episode $G \in \sspace$ and a sequence $s = s_1\cdots s_L$, let $k$ be the
smallest index such that $s_k = \lab{v}$, where $v$ is a source node in $G$.
Then $s$ covers $G$ if and only if $s[k + 1, L]$ covers $G - v$.
\end{lemma}

\begin{proof}
Let $f$ be a mapping from $V(G)$ to $s$. If $k$ is not already used by $f$,
then we can remap a source node $v$ to $k$. As $k$ is the smallest entry used by $f$,
the remaining map is a valid mapping for $G - v$ in $s[k + 1, L]$.
\end{proof}

Lemma~\ref{lem:greedy} says that to test whether a sequence $s$ covers $G$, it
is sufficient to greedily find entries from $s$ corresponding to the sources of
$G$, removing those sources as we move along. Let us denote such a mapping by
$g(G; s)$. Note that if the episode is not strict, then we can have two source
nodes with the same label, in which case it might be that Lemma~\ref{lem:greedy}
holds for one of the sources but not for the other. Since we cannot know in advance
which node to choose, this ambiguity would make the greedy approach less efficient.

\begin{lemma}
\label{lem:greedyminimal}
Let $G \in \sspace$ be an episode and let $s$ be a sequence covering $G$. Let $m = g(G; s)$
and let $[a, b]$ be the first minimal window of $G$ in $s$. Then $\max m = b$.
\end{lemma}

\begin{proof}
Since $s[1, b]$ covers $G$, there is a mapping $m' = g(G; s[1, b])$. Since $[a, b]$
is the first minimal window, we must have $\max m' = b$. Since $m$ and $m'$ are constructed
in a greedy fashion, we must have $m = m'$.
\end{proof}

Lemma~\ref{lem:greedyminimal} states that if we evaluate $m = g(G; s[k, L])$
for each $k = 1, \ldots, K$, store the intervals $W = [\min m, \max m]$, and
for each pair of windows $[a_1, b], [a_2, b] \in W$ remove $[a_2, b]$, then $W$
will contain the minimal windows. Evaluating $g(G; s)$ can be done in
polynomial time, so this approach takes polynomial time.
The algorithm for discovering minimal windows is given in Algorithm~\ref{alg:findwindows}.
To fully describe the algorithm we need the following definition.

\begin{definition}
An edge $(v, w)$ in an episode $G \in \sspace$ is called a \emph{skeleton edge} if there is
no node $u$ such that $(v, u, w)$ is a path in $G$. If $v$ and $w$ have different
labels, we call the edge $(v, w)$ a \emph{proper skeleton edge}.
\end{definition}

\begin{algorithm}[htb!]
\capstart
\Input{an episode $G \in \sspace$} 
\Input{set of minimal windows $W$}

	$W \define \emptyset$\;
	$Q \define \emptyset$\;
	\ForEach{$v \in V(G)$} {
		$f(v) \define$ first $i$ such that $s_i = \lab{v}$\;
		$b(v) \define -\infty$\;
		add $v$ into $Q$\;
	}

	\While {\True} {
		\tcpas{Make sure that $f$ honors the edges in $G$}
		\While {$Q$ is not empty} {\nllabel{alg:fw:honor1}
			$v \define $ first element in $Q$\;
			remove $v$ from $Q$\;
			$f(v) \define$ the smallest $i > b(v)$ such that $s_i = \lab{v}$\;\nllabel{alg:fw:search}
			\lIf{$f(v) = \Null$} {\Return $W$}
			\ForEach {skeleton edge $(v, w) \in E(G)$} {
				$b(w) \define \max\fpr{b(w), f(v)}$\;
				\If {$b(w) \geq f(w)$ \AND $w \notin Q$} {
					add $w$ into $Q$\;\nllabel{alg:fw:honor2}
				}
			}
		}
		\If {$\max f - \min f < \rho$} {
			\If {$\max f = b$ such that $[a, b]$ is the last entry in $W$} {
				delete the last entry from $W$\;
			}
			add $[\min f, \max f]$ to $W$\;
		}
		$v \define $ node with the smallest $f(v)$\;\nllabel{alg:fw:minimal}
		$b(v) \define f(v)$\;\nllabel{alg:fw:next}
		add $v$ to $Q$\;
	}
\caption{\textsc{FindWindows}. An algorithm for finding minimal windows of $G$ from $s$. The parameter $\rho$ is the maximal size of the window.}
\label{alg:findwindows}
\end{algorithm}

\begin{theorem}
Algorithm \textsc{FindWindows} discovers all minimal windows.
\end{theorem}
\begin{proof}
We will prove the correctness of the algorithm in several steps.  First, we
will show that the loop \ref{alg:fw:honor1}--\ref{alg:fw:honor2} guarantees
that $f$ is truly a valid mapping. To see this, note that the loop
upholds the following invariants: For every $v \notin Q$, we have $f(v) > b(v)$
and for each skeleton edge $(v, w)$, we have $b(w) \neq -\infty$. Also, for
every non-source node $w$, $b(w) = \max \set{f(v) \mid (v, w) \text { is a
skeleton edge}}$ or $b(w) = -\infty$.

Thus when we leave the loop, that is, $Q = \emptyset$, then $f$ is a valid
mapping for $G$. Moreover, since we select the smallest $f(v)$ during
Line~\ref{alg:fw:search}, we can show using recursion that $f = g(G; s[k + 1, L])$,
where $k = \min b(v)$, where the min ranges over all source nodes.
Once $f$ is discovered, the next new mapping should be contained in $s[\min f + 1, L]$.
This is exactly what is done on Line~\ref{alg:fw:next} by making $b(v)$ equal to $f(v)$.
\end{proof}

The advantage of this approach is that we can now optimize the search on
Line~\ref{alg:fw:search} by starting the search from the previous index $f(v)$
instead of starting from the start of the sequence. The following lemma, implied directly
by the greedy procedure, allows this.

\begin{lemma}
Let $G \in \sspace$ be an episode and let $s$ be a sequence covering $G$. Let $k$ be an index and
assume that $s[k, L]$ covers $G$.  Set $m_1 = g(G; s)$ and $m_2 = g(G; s[k + 1, L])$.
Then $m_2(v) \geq m_1(v)$ for every $v \in V(G)$.
\end{lemma}

Let us now analyze the complexity of \textsc{FindWindows}. Assume that the
input episode $G$ has $N$ nodes and the input sequence $s$ has $L$ symbols. Let
$K$ be the maximal number of nodes in $G$ sharing the same label. Let $D$
be the maximal number of outgoing skeleton edges in a single node.  Note that
each event in $s$ is visited during Line~\ref{alg:fw:search} $K$ times, at
maximum. Each visit may require $D$ updates of $b(w)$, at maximum. The only
remaining non-trivial step is finding the minimal source
(Line~\ref{alg:fw:minimal}), which can be done with a heap in $O(\log N)$
time. This brings the total evaluation time to $O((D + 1)KL + L\log N)$.
Note that for parallel episodes $D = 0$ and for serial episodes $D = 1$, and
for such episodes we can further optimize the algorithm such that each event is
visited only twice, thus making the running time to be in $O(L\log N)$ for parallel episodes and in $O(L)$ for serial episodes.
Since there can be only $L$ minimal windows, the additional space complexity
for the algorithm is in $O(L)$.

Given the set of minimal windows, we can now compute the node closure by
testing which nodes occur in all minimal windows. This can be done in $O(L)$
time. To compute the edge closure we need to perform $N^2$ calls of
\textsc{FindWindows}. This can be greatly optimized by sieving multiple edges
simultaneously based on valid mappings $f$ discovered during 
\textsc{FindWindows}. In addition, we can compute both frequency measures
from minimal windows in $O(L)$ time.

We will now turn to discovering $f$-closed episodes.
Note that, unlike the $i$-closure, we do not define an $f$-closure of an
episode at all. As shown in Section~\ref{sec:introduction}, such an $f$-closure
would not necessarily be unique. However, the set of $i$-closed episodes
will contain all $f$-closed episodes.

\begin{theorem}
\label{prop:frin}
An $f$-closed episode is always $i$-closed.
\end{theorem}

\begin{proof}
Let $G \in \sspace$ be an episode and define $H = \iclosure{G}$. Let $W$ be a set of all
minimal windows of $s$ that cover $G$ and let $V$ be the set of all minimal
windows that cover $H$. Let $w = [a, b] \in W$ be a minimal window of $G$ and
let $f$ be a valid mapping from $G$ to $s[a, b]$.  Any new node added in $H$
must occur in $s[a, b]$ and $f$ can be extended to $H$ such that it honours all
edges in $H$. Hence $s[a, b]$ covers $H$. It is also a minimal window for $H$
because otherwise it would violate the minimality for $G$. Hence $w \in V$.
Now assume that $v = [a', b'] \in V$. Obviously, $s[a', b']$ covers $G$, so
there must be a minimal window $w = [a, b] \in W$. Using the first argument we
see that $w \in V$. Since $V$ contains only minimal windows we conclude that $v = w$.
Hence $W = V$. 
A sequence covers an episode if and only if the sequence contains a minimal
window of the episode.  Thus, by definition, $V = W$ implies that $\freqf{G} =
\freqf{H}$  and $\freqm{G} = \freqm{H}$.

Assume now that $G$ is not $i$-closed, then $G \neq H$, and since $\freq{G} =
\freq{H}$, $G$ is not $f$-closed either. This completes the proof.
\end{proof}

A na\"{i}ve approach to extract $f$-closed episodes would be to compare each pair
of instance-closed episodes $G$ and $H$ and if $\freq{G}= \freq{H}$  and $G
\prec H$, remove $G$ from the output. This approach can be considerably sped up
by realising that we need only to test episodes with identical nodes and
episodes of form $G - V$, where $V$ is a subset of $V(G)$. The pseudo-code is given in 
Algorithm~\ref{alg:fclosure}. The algorithm can be further sped up by exploiting
the subset relationship between the episodes.
Our experiments demonstrate that this
comparison is feasible in practice.

\begin{algorithm}[htb!]
\capstart
	\Input{all $i$-closed episodes $\efam{G}$}
	\Out{all $f$-closed episodes}
	\ForEach{$G \in \efam{G}$} {
		\ForEach{$H \in \efam{G}$ with $V(G) = V(H)$, $H \neq G$} {
			\If{$G \prec H$ \AND $\freq{G} = \freq{H}$} {
				Mark $G$\;
			}
			\If{$H \prec G$ \AND $\freq{G} = \freq{H}$} {
				Mark $H$\;
			}
		}
		\ForEach{$V \subset V(G)$} {
			$F \define G - V$\;
			\ForEach{$H \in \efam{G}$, with $V(F) = V(H)$} {
				\If{$H \preceq F$ \AND $\freq{G} = \freq{H}$} {
					Mark $H$\;
				}
			}
		}
	}
	\Return all unmarked episodes from $\efam{G}$\; 
	
\caption{\textsc{F-Closure}. Postprocessing for computing $f$-closed episodes from $i$-closures.}
\label{alg:fclosure}
\end{algorithm}

\section{Generating Transitively Closed Candidate Episodes}
\label{sec:candidate}

In this section we define an algorithm, \textsc{GenerateCandidate}, which
generates the candidate episodes from the episodes discovered previously.
The difficulty of the algorithm is that we need to make sure 
that the candidates are transitively
closed.

Let $G \in \sspace$ be an episode. It is easy to see that if we remove a
proper skeleton edge $e$ from $G$, then the resulting episode $G - e$ will be
in $\sspace$, that is, transitively closed and strict.  We can reverse this property in order to generate
candidates: Let $G \in \sspace$ be a previously discovered episode, add
an edge $e$ and verify that the new episode is transitively closed.  However,
we can improve on this na\"{i}ve approach with the following theorem describing
the sufficient and necessary condition for an episode to be transitively
closed.

\begin{theorem}
\label{prop:closedadd}
Let $G \in \sspace$ be an episode and let $e = (x, y)$ be an edge not in
$E(G)$.  Let $H = G + e$. Assume that $H$ is a DAG.  Then $H \in \sspace$
if and only if there is an edge in $G$ from $x$ to every child
of $y$ and from every parent of $x$ to $y$.
\end{theorem}

\begin{proof}
The 'only if' part follows directly from the definition of transitive
closure.  To prove the 'if' part, we will use induction.  Let $u$ be an ancestor
node of $v$ in $H$.  Then there is a path from $u$ to $v$ in $H$.  If the path does
not use edge $e$, then, since $G$ is transitively closed, $(u, v) \in E(G)$ and
hence $(u, v) \in E(H)$. Assume now that the path uses $e$. If $v = y$, then $u$
must be a parent of $y$ in $G$, since $G$ is transitively closed, so the condition implies that $(u, v) \in E(G) \subset
E(H)$. Assume that $v$ is a proper descendant of $y$ in $H$. To prove the first step in the
induction, assume that $u = x$, then again $(u, v) \in E(G)$. To prove the induction step, let $w$ be the
next node along the path from $u$ to $v$ in $H$.  Assume inductively that $(w, v) \in E(G)$.
Then the path $(u, w, v)$ occurs in $G$, so $(u, v) \in E(G)$, which completes
the proof.
\end{proof}

We now show when we can join two episodes to obtain a candidate episode.  Since
our nodes are ordered, we can also order the edges using a lexicographical
order.  Given an episode $G$ we define $\last{G}$ to be the last proper
skeleton edge in $G$. The next theorem shows the necessary conditions for the
existence of the candidate episode.

\begin{theorem}
\label{thr:generate}
Let $H \in \sspace$ be an episode with $N + 1$ edges. Then either there
are two episodes, $G_1, G_2 \in \sspace$, with identical nodes to $H$, such that
\begin{itemize}
\item $G_1$ and $G_2$ share $N - 1$ edges, 
\item $e_1 = \last{G_1} \notin E(G_2)$,
\item $e_2 > e_1$ and $H = G_1 + e_2$, where $e_2$ is the unique edge in $E(G_2)$ and not in $E(G_1)$\footnote{In~\cite{tatti:10:mining} we incorrectly stated that $e_2 = \last{G_2}$. Generally, this is not the case.} 
\end{itemize}
or $\last{G}$ is no longer a skeleton edge in $H$, where $G \in \sspace$, $G = H - \last{H}$.
\end{theorem}

We will refer to the two cases as Case A and Case B, respectively.

\begin{proof}
Let $e_2 = \last{H}$ and define $G_1 = H - e_2$. If
$e_1$ is not a skeleton edge in $H$, then by setting $G = G_1$, the theorem
holds.  Assume that $e_1$ is a skeleton edge in $H$. Define $G_2 = H - e_1$.
Note that $G_1$ and $G_2$ share $N - 1$ edges. Also note that $\last{G_1} = e_1
\notin E(G_2)$, $\set{e_2} = E(G_2) - E(G_1)$, and $H = G_1 + e_2$. Since $e_2
= \last{H}$ it must be that $e_2 > e_1$.
\end{proof}

Theorem~\ref{thr:generate} gives us means to generate episodes. Let us first
consider Case A in Theorem~\ref{thr:generate}. To generate $H$ we
simply find all pairs of episodes $G_1$ and $G_2$ such that the conditions of Case A in
Theorem~\ref{thr:generate} hold. When combining $G_1$ and $G_2$ we need to test
whether the resulting episode is transitively closed.

\begin{theorem}
\label{thr:joinsafe}
Let $G_1, G_2 \in \sspace$ be two episodes with identical nodes
and $N$ edges. Assume that $G_1$ and $G_2$ share $N - 1$ mutual edges. Let $e_1 =
(x_1, y_1) \in E(G_1) - E(G_2)$ be the unique edge of $G_1$ and let $e_2 =
(x_2, y_2) \in E(G_2) - E(G_1)$ be the unique edge of $G_2$. Let $H = G_1 +
e_2$.  Assume that $H$ has no cycles. Then $H \in \sspace$ if
and only if one of the following conditions is true:
\begin{enumerate}
\item $x_1 \neq y_2$ and $x_2 \neq y_1$.
\item $x_1 \neq y_2$, $x_2 = y_1$,  and $(x_1, y_2)$ is an edge in $G_1$.
\item $x_1 = y_2$, $x_2 \neq y_1$,  and $(x_2, y_1)$ is an edge in $G_1$.
\end{enumerate}
Moreover, if $H \in \sspace$, then $e_1$ and $e_2$ are both skeleton edges in $H$.
\end{theorem}

\begin{proof}
We will first show that $e_1$ is a skeleton edge in $H$ if $H \in \sspace$. If it is not, then
there is a path from $x_1$ to $y_1$ in $H$ not using $e_1$. The edges along
this path also occur in $G_2$, thus forcing $e_1$ to be an edge in $G_2$, which
is a contradiction. A similar argument holds for $e_2$.

The ''only if'' part is trivial so we only prove the ''if'' part using Theorem~\ref{prop:closedadd}.

Let $v$ be a child of $y_2$ in $G_1$ and $f = (y_2, v)$ an edge in $G_1$.

If the first or second condition holds, then $x_1 \neq y_2$, and consequently $f
\neq e_1$, so $f \in G_2$. The path $(x_2, y_2, v)$ connects $x_2$ and $v$ in $G_2$ so
there must be an edge $h = (x_2, v)$ in $G_2$. Since $h \neq e_2$, $h$ must also
occur in $G_1$. 
If the third condition holds, it may be the case that $f = e_1$ (if not, then
we can use the previous argument). But in such a case $v = y_1$ and edge $h =
(x_2, y_1)$ occurs in $G_1$. 

If now $u$ is a parent of $x_2$ in $G_1$, we can make a similar argument that
$u$ and $y_2$ are connected, so Theorem~\ref{prop:closedadd} now implies
that $H$ is transitively closed.
\end{proof}

Theorem~\ref{thr:joinsafe} allows us to handle Case A of
Theorem~\ref{thr:generate}.  To handle Case B, we simply take an
episode $G$ and try to find all edges $e_2$ such that $\last{G + e_2} = e_2$
and $\last{G}$ is no longer a skeleton edge in $G + e_2$. The conditions for
this are given in the next theorem.

\begin{theorem}
\label{thr:extendhide}
Let $G \in \sspace$ be an episode, let $e_1 = (x_1, y_1)$ be a skeleton
edge of $G$, and let $e_2 = (x_2, y_2)$ be an edge not occurring in $G$ and
define $H = G + e_2$. Assume $H \in \sspace$ and that $e_1$
is not a skeleton edge in $H$. Then either $y_2 = y_1$ and $(x_1, x_2)$ is a
skeleton edge in $G$ or $x_1 = x_2$ and $(y_2, y_1)$ is a skeleton edge in $G$.
\end{theorem}

\begin{proof}
Assume that $e_1$ is no longer a skeleton edge in $H$, then there is a path of
skeleton edges going from $x_1$ to $y_1$ in $H$ not using $e_1$. The path must
use $e_2$, otherwise we have a contradiction. The theorem will follow if we can
show that the path must have exactly two edges. Assume otherwise.  Assume, for
simplicity, that edge $e_2$ does not occur first in the path and let $z$ be
the node before $x_2$ in the path. Then we can build a new path by replacing
edges $(z, x_2)$ and $e_2$ with $(z, y_2)$. This path does not use $e_2$,
hence it occurs in $G$. If $z \neq x_1$ or $y_1 \neq y_2$, then the path makes $e_1$ a non-skeleton edge in $G$, which is a
contradiction. If $e_2$ is the first edge in the path, we can select the next
node after $y_2$ and repeat the argument.
\end{proof}

\begin{example}
\label{ex:join1}
Consider the episodes given in Figure~\ref{fig:join1}. Episodes $G_1$ and $G_2$ satisfy
the second condition in Theorem~\ref{thr:joinsafe}, hence the resulting episode $H_1$,
is transitively closed. On the other hand, combining $G_1$ and $G_3$ leads to $H'_1$, an episode
that is not transitively closed since edge $(a, d)$ is missing.
Finally, $G_4$ and $(c, a)$ satisfy Theorem~\ref{thr:extendhide} and
generate $H_2$.
\end{example}

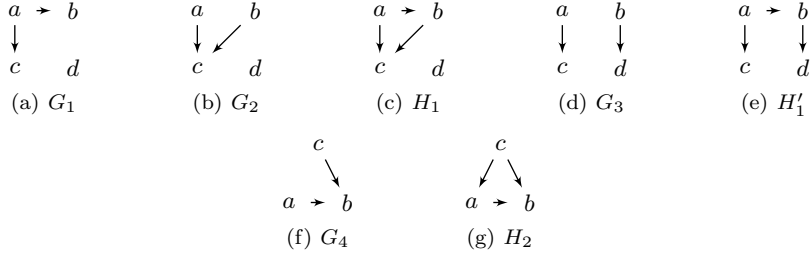
\begin{figure}[htb!]
\centering
\subfigure[$G_1$]{\begin{minipage}[b]{2.3cm}\centering\normalsize\begin{tikzpicture}[>=latex',line join=bevel,scale=0.3]\pgfsetlinewidth{0.5bp}\pgfsetcolor{black}
  \draw [->] (27bp,71.831bp) .. controls (27bp,64.131bp) and (27bp,54.974bp)  .. (27bp,36.413bp);
  \draw [->] (54bp,90bp) .. controls (56.615bp,90bp) and (59.229bp,90bp)  .. (71.93bp,90bp);
\begin{scope}
  \definecolor{strokecol}{rgb}{0.0,0.0,0.0};
  \pgfsetstrokecolor{strokecol}
  \draw (27bp,90bp) node {$a$};
\end{scope}
\begin{scope}
  \definecolor{strokecol}{rgb}{0.0,0.0,0.0};
  \pgfsetstrokecolor{strokecol}
  \draw (27bp,18bp) node {$c$};
\end{scope}
\begin{scope}
  \definecolor{strokecol}{rgb}{0.0,0.0,0.0};
  \pgfsetstrokecolor{strokecol}
  \draw (99bp,90bp) node {$b$};
\end{scope}
\begin{scope}
  \definecolor{strokecol}{rgb}{0.0,0.0,0.0};
  \pgfsetstrokecolor{strokecol}
  \draw (99bp,18bp) node {$d$};
\end{scope}\end{tikzpicture}\end{minipage}}
\subfigure[$G_2$]{\begin{minipage}[b]{2.3cm}\centering\normalsize\begin{tikzpicture}[>=latex',line join=bevel,scale=0.3]\pgfsetlinewidth{0.5bp}\pgfsetcolor{black}
  \draw [->] (80.831bp,71.831bp) .. controls (72.285bp,63.285bp) and (61.944bp,52.944bp)  .. (45.413bp,36.413bp);
  \draw [->] (27bp,71.831bp) .. controls (27bp,64.131bp) and (27bp,54.974bp)  .. (27bp,36.413bp);
\begin{scope}
  \definecolor{strokecol}{rgb}{0.0,0.0,0.0};
  \pgfsetstrokecolor{strokecol}
  \draw (27bp,90bp) node {$a$};
\end{scope}
\begin{scope}
  \definecolor{strokecol}{rgb}{0.0,0.0,0.0};
  \pgfsetstrokecolor{strokecol}
  \draw (27bp,18bp) node {$c$};
\end{scope}
\begin{scope}
  \definecolor{strokecol}{rgb}{0.0,0.0,0.0};
  \pgfsetstrokecolor{strokecol}
  \draw (99bp,90bp) node {$b$};
\end{scope}
\begin{scope}
  \definecolor{strokecol}{rgb}{0.0,0.0,0.0};
  \pgfsetstrokecolor{strokecol}
  \draw (99bp,18bp) node {$d$};
\end{scope}\end{tikzpicture}\end{minipage}}
\subfigure[$H_1$]{\begin{minipage}[b]{2.3cm}\centering\normalsize\begin{tikzpicture}[>=latex',line join=bevel,scale=0.3]\pgfsetlinewidth{0.5bp}\pgfsetcolor{black}
  \draw [->] (54bp,90bp) .. controls (56.615bp,90bp) and (59.229bp,90bp)  .. (71.93bp,90bp);
  \draw [->] (80.831bp,71.831bp) .. controls (72.285bp,63.285bp) and (61.944bp,52.944bp)  .. (45.413bp,36.413bp);
  \draw [->] (27bp,71.831bp) .. controls (27bp,64.131bp) and (27bp,54.974bp)  .. (27bp,36.413bp);
\begin{scope}
  \definecolor{strokecol}{rgb}{0.0,0.0,0.0};
  \pgfsetstrokecolor{strokecol}
  \draw (27bp,90bp) node {$a$};
\end{scope}
\begin{scope}
  \definecolor{strokecol}{rgb}{0.0,0.0,0.0};
  \pgfsetstrokecolor{strokecol}
  \draw (27bp,18bp) node {$c$};
\end{scope}
\begin{scope}
  \definecolor{strokecol}{rgb}{0.0,0.0,0.0};
  \pgfsetstrokecolor{strokecol}
  \draw (99bp,90bp) node {$b$};
\end{scope}
\begin{scope}
  \definecolor{strokecol}{rgb}{0.0,0.0,0.0};
  \pgfsetstrokecolor{strokecol}
  \draw (99bp,18bp) node {$d$};
\end{scope}\end{tikzpicture}\end{minipage}}
\subfigure[$G_3$]{\begin{minipage}[b]{2.3cm}\centering\normalsize\begin{tikzpicture}[>=latex',line join=bevel,scale=0.3]\pgfsetlinewidth{0.5bp}\pgfsetcolor{black}
  \draw [->] (27bp,71.831bp) .. controls (27bp,64.131bp) and (27bp,54.974bp)  .. (27bp,36.413bp);
  \draw [->] (99bp,71.831bp) .. controls (99bp,64.131bp) and (99bp,54.974bp)  .. (99bp,36.413bp);
\begin{scope}
  \definecolor{strokecol}{rgb}{0.0,0.0,0.0};
  \pgfsetstrokecolor{strokecol}
  \draw (27bp,90bp) node {$a$};
\end{scope}
\begin{scope}
  \definecolor{strokecol}{rgb}{0.0,0.0,0.0};
  \pgfsetstrokecolor{strokecol}
  \draw (27bp,18bp) node {$c$};
\end{scope}
\begin{scope}
  \definecolor{strokecol}{rgb}{0.0,0.0,0.0};
  \pgfsetstrokecolor{strokecol}
  \draw (99bp,90bp) node {$b$};
\end{scope}
\begin{scope}
  \definecolor{strokecol}{rgb}{0.0,0.0,0.0};
  \pgfsetstrokecolor{strokecol}
  \draw (99bp,18bp) node {$d$};
\end{scope}\end{tikzpicture}\end{minipage}}
\subfigure[$H'_1$]{\begin{minipage}[b]{2.3cm}\centering\normalsize\begin{tikzpicture}[>=latex',line join=bevel,scale=0.3]\pgfsetlinewidth{0.5bp}\pgfsetcolor{black}
  \draw [->] (54bp,90bp) .. controls (56.615bp,90bp) and (59.229bp,90bp)  .. (71.93bp,90bp);
  \draw [->] (27bp,71.831bp) .. controls (27bp,64.131bp) and (27bp,54.974bp)  .. (27bp,36.413bp);
  \draw [->] (99bp,71.831bp) .. controls (99bp,64.131bp) and (99bp,54.974bp)  .. (99bp,36.413bp);
\begin{scope}
  \definecolor{strokecol}{rgb}{0.0,0.0,0.0};
  \pgfsetstrokecolor{strokecol}
  \draw (27bp,90bp) node {$a$};
\end{scope}
\begin{scope}
  \definecolor{strokecol}{rgb}{0.0,0.0,0.0};
  \pgfsetstrokecolor{strokecol}
  \draw (27bp,18bp) node {$c$};
\end{scope}
\begin{scope}
  \definecolor{strokecol}{rgb}{0.0,0.0,0.0};
  \pgfsetstrokecolor{strokecol}
  \draw (99bp,90bp) node {$b$};
\end{scope}
\begin{scope}
  \definecolor{strokecol}{rgb}{0.0,0.0,0.0};
  \pgfsetstrokecolor{strokecol}
  \draw (99bp,18bp) node {$d$};
\end{scope}\end{tikzpicture}\end{minipage}}
\subfigure[$G_4$]{\begin{minipage}[b]{2.3cm}\centering\normalsize\begin{tikzpicture}[>=latex',line join=bevel,scale=0.3]\pgfsetlinewidth{0.5bp}\pgfsetcolor{black}
  \draw [->] (54bp,18bp) .. controls (56.615bp,18bp) and (59.229bp,18bp)  .. (71.93bp,18bp);
  \draw [->] (72.084bp,71.831bp) .. controls (76.061bp,63.877bp) and (80.815bp,54.369bp)  .. (89.793bp,36.413bp);
\begin{scope}
  \definecolor{strokecol}{rgb}{0.0,0.0,0.0};
  \pgfsetstrokecolor{strokecol}
  \draw (27bp,18bp) node {$a$};
\end{scope}
\begin{scope}
  \definecolor{strokecol}{rgb}{0.0,0.0,0.0};
  \pgfsetstrokecolor{strokecol}
  \draw (63bp,90bp) node {$c$};
\end{scope}
\begin{scope}
  \definecolor{strokecol}{rgb}{0.0,0.0,0.0};
  \pgfsetstrokecolor{strokecol}
  \draw (99bp,18bp) node {$b$};
\end{scope}\end{tikzpicture}\end{minipage}}
\subfigure[$H_2$]{\begin{minipage}[b]{2.3cm}\centering\normalsize\begin{tikzpicture}[>=latex',line join=bevel,scale=0.3]\pgfsetlinewidth{0.5bp}\pgfsetcolor{black}
  \draw [->] (53.916bp,71.831bp) .. controls (49.939bp,63.877bp) and (45.185bp,54.369bp)  .. (36.207bp,36.413bp);
  \draw [->] (54bp,18bp) .. controls (56.615bp,18bp) and (59.229bp,18bp)  .. (71.93bp,18bp);
  \draw [->] (72.084bp,71.831bp) .. controls (76.061bp,63.877bp) and (80.815bp,54.369bp)  .. (89.793bp,36.413bp);
\begin{scope}
  \definecolor{strokecol}{rgb}{0.0,0.0,0.0};
  \pgfsetstrokecolor{strokecol}
  \draw (27bp,18bp) node {$a$};
\end{scope}
\begin{scope}
  \definecolor{strokecol}{rgb}{0.0,0.0,0.0};
  \pgfsetstrokecolor{strokecol}
  \draw (63bp,90bp) node {$c$};
\end{scope}
\begin{scope}
  \definecolor{strokecol}{rgb}{0.0,0.0,0.0};
  \pgfsetstrokecolor{strokecol}
  \draw (99bp,18bp) node {$b$};
\end{scope}\end{tikzpicture}\end{minipage}}
\caption{Toy episodes for Example~\ref{ex:join1}.}
\label{fig:join1}
\end{figure}

We can now combine the preceding theorems
into the \textsc{GenerateCandidate} algorithm given in
Algorithm~\ref{alg:generatecandidates}. We will first generate candidates by
combining episodes from the previous rounds using Theorem~\ref{thr:joinsafe}.
Secondly, we use Theorem~\ref{thr:extendhide} and for each episode from the
previous rounds we add edges such that the last proper skeleton edge is no longer a
skeleton edge in the candidate.

\begin{algorithm}[htb!]
\capstart
\Input{a collection of previously discovered episodes $\efam{G} \subset \sspace$ with $N$ edges}
\Out{$\efam{P} \subset \sspace$, a collection of candidate episodes with $N + 1$ edges}
	$\efam{P} \define \emptyset$\;

	\ForEach{$G \in \efam{G}$, $G$ has no proper edges} {
		\ForEach{$x, y \in V(G)$, $\lab{x} \neq \lab{y}$} {
			add $G + (x, y)$ to $\efam{P}$\;
		}
	}
	\ForEach{$G_1 \in \efam{G}$, $G_1$ has proper edges} {
		$e_1 = (x_1, y_1) \define \last{G_1}$\;
		\tcpas{Case where $e_1$ remains a skeleton edge.}
		$\efam{H} \define \set{H \in \efam{G} | V(G) = V(H), \abs{E(H) \cap E(G)} = N - 1, \set{e_2} = E(H) - E(G), e_2 > e_1}$\; 
		\ForEach{$G_2$ in $\efam{H}$} {
			$\set{e_2} \define E(H) - E(G)$\; 
			\If{$G_1$ and $G_2$ satisfy Theorem~\ref{thr:joinsafe} \AND $e_2 \notin \iclosure{G_1}$ } {
				add $G_1 + e_2$ to $\efam{P}$\;
			}
		}
		\tcpas{Case where $e_1$ does not remain a skeleton edge.}
		\ForEach{$f = (x_1, x_2)$ skeleton edge in $G_1$ such that $x_2 \neq y_1$} {
			$e_2 \define (x_2, y_1)$\;
			\If {$e_2 \notin \iclosure{G_1}$} {
				$H \define G_1 + e_2$\;
				\If {$e_2 = \last{H}$ \AND $H$ is transitively closed (use Theorem~\ref{prop:closedadd})} { 
					add $H$ to $\efam{P}$\;
				}
			}
		}
		\ForEach{$f = (y_2, y_1)$ skeleton edge in $G_1$ such that $y_2 \neq x_1$} {
			$e_2 \define (x_1, y_2)$\;
			\If {$e_2 \notin \iclosure{G_1}$} {
				$H \define G_1 + e_2$\;
				\If {$e_2 = \last{H}$ \AND $H$ is transitively closed (use Theorem~\ref{prop:closedadd})} { 
					add $H$ to $\efam{P}$\;
				}
			}
		}
	}
	\Return $\efam{P}$\;
\caption{\textsc{GenerateCandidate}. Generates candidate episodes from the
previously discovered episodes.}
\label{alg:generatecandidates}
\end{algorithm}

Note that Algorithm~\ref{alg:generatecandidates} will not generate duplicates.
To see this, first note that if $H$ is generated using Case A, then the two
episodes $G_1$ and $G_2$ generating $H$ are unique.  That is, it must be that
$G_1 = H - \last{H}$ and $G_2 = H - \last{G_1}$. In other words, $G_1$ and
$G_2$ will produce a unique $H$. In Case B, the episode $G_1$ is also
unique, namely, we must have $G_1 = H - \last{H}$. Thus, each $G_1$ will
produce a unique $H$.  Finally, note that $\last{G_1}$, according to
Theorem~\ref{thr:joinsafe} is always a skeleton edge in $H$ in Case A
and we demand that $\last{G_1}$ is not a skeleton edge in Case B.
Hence, Case A and Case B will never generate the same episode.

\section{Algorithm for Discovering Closed Episodes}
\label{sec:algorithm}

In this section we will complete our mining algorithm for discovering closed
episodes. First, we present an algorithm for testing the episodes and a more
detailed version of Algorithm~\ref{alg:sketch}. Then we present an
algorithm for adding the needed intermediate episodes and prove the correctness
of the algorithm.

\subsection{Detailed version of the algorithm}

We begin by describing the test subroutine that is done for each candidate episode.
Following the level-wise discovery, before computing the
frequency of the episode, we need to test that all its subepisodes are discovered. 
It turns out that using transitively closed episodes will guarantee the strongest
conditions for an episode to pass to the frequency computation stage.

\begin{corollary}[of Theorem~\ref{thr:sub}]
\label{cor:optimal}
Let $G \in \sspace$ be an episode. Let $e$ be a proper skeleton edge of $G$. If $H$ is an episode
obtained by removing $e$ from $G$, then there exists no episode $H_1 \in \sspace$, such that $H\prec H_1 \prec G$.
\end{corollary}

If $e$ is a proper skeleton edge of an episode $G \in \sspace$, then $G - e \in \sspace$. Thus, for $G$ to be frequent,
$G - e$ had to be discovered previously. This is the first test in
\textsc{TestCandidate} (given in Algorithm~\ref{alg:testcandidate}).  In
addition, following the level-wise approach for mining closed
patterns~\cite{pasquier:99:discovering}, we test that $G$ is not a subepisode
of $\iclosure{G - e}$, and if it is, then we can discard $G$.

The second test involves testing whether $G - v$, where $v$ is a node in
$G$, has also been discovered. Note that $G - v$ has fewer nodes than $G$ so, if
$G$ is frequent, we must have discovered $G - v$. Not all nodes need to be
tested.  If a node $v$ has an adjacent proper skeleton edge, say $e$, then the episode
$G - e$ has a frequency lower than or equal to that of $G - v$. Since we have already
tested $G - e$ we do not need to test $G - v$. Consequently, we need to test
only those nodes that have no proper skeleton edges. This leads us to the second test
in \textsc{TestCandidate}. Note that these nodes will either have no edges,
or will have edges to the nodes having the same label.
If both tests are passed we
test the candidate episode for frequency.

\begin{example}
\label{ex:test}
Consider the episodes given in Figure~\ref{fig:test}. When testing the candidate
episode  $H$, \textsc{TestCandidate} will test for the existence of three
episodes. Episodes $G_1$ and $G_2$ are tested when either edge is removed and
$G_3$ is tested when node labelled $d$ is removed.
\end{example}

\begin{figure}[htb!]
\centering
\subfigure[$H$]{\begin{minipage}[b]{2.5cm}\centering\normalsize\begin{tikzpicture}[>=latex',line join=bevel,scale=0.3]\pgfsetlinewidth{0.5bp}\pgfsetcolor{black}
  \draw [->] (53.916bp,71.831bp) .. controls (49.939bp,63.877bp) and (45.185bp,54.369bp)  .. (36.207bp,36.413bp);
  \draw [->] (72.084bp,71.831bp) .. controls (76.061bp,63.877bp) and (80.815bp,54.369bp)  .. (89.793bp,36.413bp);
\begin{scope}
  \definecolor{strokecol}{rgb}{0.0,0.0,0.0};
  \pgfsetstrokecolor{strokecol}
  \draw (63bp,90bp) node {$a$};
\end{scope}
\begin{scope}
  \definecolor{strokecol}{rgb}{0.0,0.0,0.0};
  \pgfsetstrokecolor{strokecol}
  \draw (99bp,18bp) node {$c$};
\end{scope}
\begin{scope}
  \definecolor{strokecol}{rgb}{0.0,0.0,0.0};
  \pgfsetstrokecolor{strokecol}
  \draw (27bp,18bp) node {$b$};
\end{scope}
\begin{scope}
  \definecolor{strokecol}{rgb}{0.0,0.0,0.0};
  \pgfsetstrokecolor{strokecol}
  \draw (135bp,90bp) node {$d$};
\end{scope}\end{tikzpicture}\end{minipage}}
\subfigure[$G_1$]{\begin{minipage}[b]{2.5cm}\centering\normalsize\begin{tikzpicture}[>=latex',line join=bevel,scale=0.3]\pgfsetlinewidth{0.5bp}\pgfsetcolor{black}
  \draw [->] (53.916bp,71.831bp) .. controls (49.939bp,63.877bp) and (45.185bp,54.369bp)  .. (36.207bp,36.413bp);
\begin{scope}
  \definecolor{strokecol}{rgb}{0.0,0.0,0.0};
  \pgfsetstrokecolor{strokecol}
  \draw (63bp,90bp) node {$a$};
\end{scope}
\begin{scope}
  \definecolor{strokecol}{rgb}{0.0,0.0,0.0};
  \pgfsetstrokecolor{strokecol}
  \draw (99bp,18bp) node {$c$};
\end{scope}
\begin{scope}
  \definecolor{strokecol}{rgb}{0.0,0.0,0.0};
  \pgfsetstrokecolor{strokecol}
  \draw (27bp,18bp) node {$b$};
\end{scope}
\begin{scope}
  \definecolor{strokecol}{rgb}{0.0,0.0,0.0};
  \pgfsetstrokecolor{strokecol}
  \draw (135bp,90bp) node {$d$};
\end{scope}\end{tikzpicture}\end{minipage}}
\subfigure[$G_2$]{\begin{minipage}[b]{2.5cm}\centering\normalsize\begin{tikzpicture}[>=latex',line join=bevel,scale=0.3]\pgfsetlinewidth{0.5bp}\pgfsetcolor{black}
  \draw [->] (72.084bp,71.831bp) .. controls (76.061bp,63.877bp) and (80.815bp,54.369bp)  .. (89.793bp,36.413bp);
\begin{scope}
  \definecolor{strokecol}{rgb}{0.0,0.0,0.0};
  \pgfsetstrokecolor{strokecol}
  \draw (63bp,90bp) node {$a$};
\end{scope}
\begin{scope}
  \definecolor{strokecol}{rgb}{0.0,0.0,0.0};
  \pgfsetstrokecolor{strokecol}
  \draw (99bp,18bp) node {$c$};
\end{scope}
\begin{scope}
  \definecolor{strokecol}{rgb}{0.0,0.0,0.0};
  \pgfsetstrokecolor{strokecol}
  \draw (27bp,18bp) node {$b$};
\end{scope}
\begin{scope}
  \definecolor{strokecol}{rgb}{0.0,0.0,0.0};
  \pgfsetstrokecolor{strokecol}
  \draw (135bp,90bp) node {$d$};
\end{scope}\end{tikzpicture}\end{minipage}}
\subfigure[$G_3$]{\begin{minipage}[b]{2.5cm}\centering\normalsize\begin{tikzpicture}[>=latex',line join=bevel,scale=0.3]\pgfsetlinewidth{0.5bp}\pgfsetcolor{black}
  \draw [->] (53.916bp,71.831bp) .. controls (49.939bp,63.877bp) and (45.185bp,54.369bp)  .. (36.207bp,36.413bp);
  \draw [->] (72.084bp,71.831bp) .. controls (76.061bp,63.877bp) and (80.815bp,54.369bp)  .. (89.793bp,36.413bp);
\begin{scope}
  \definecolor{strokecol}{rgb}{0.0,0.0,0.0};
  \pgfsetstrokecolor{strokecol}
  \draw (63bp,90bp) node {$a$};
\end{scope}
\begin{scope}
  \definecolor{strokecol}{rgb}{0.0,0.0,0.0};
  \pgfsetstrokecolor{strokecol}
  \draw (99bp,18bp) node {$c$};
\end{scope}
\begin{scope}
  \definecolor{strokecol}{rgb}{0.0,0.0,0.0};
  \pgfsetstrokecolor{strokecol}
  \draw (27bp,18bp) node {$b$};
\end{scope}\end{tikzpicture}\end{minipage}}
\caption{Toy episodes for Example~\ref{ex:test}.}
\label{fig:test}
\end{figure}
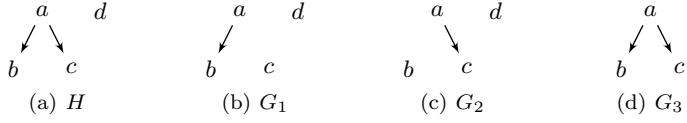

\begin{algorithm}[htb!]
\capstart
\Input{an episode $G \in \sspace$, already discovered episodes $\efam{C} \subset \sspace$}
\Out{a boolean value, \True only if all subepisodes of $G$ are frequent}
	\ForEach{proper skeleton edge $e$ in $G$} {
		\If {$G - e \notin \efam{C}$  \OR $e \in E(\iclosure{G - e})$} {\Return \False\;} 
	}
	\ForEach{$v$ in $G$ not having a proper skeleton edge} {
		\If {$G - v \notin \efam{C}$} {\Return \False\; } 
		\If {$v$ has no edges \AND there is a node $w \in V(\iclosure{G - v})$ s.t. $\lab{w} = \lab{v}$} {\Return \False\; }
	}
	\Return \True\;
\caption{\textsc{TestCandidate}. An algorithm that checks if an episode is a proper candidate.}
\label{alg:testcandidate}
\end{algorithm}

Finally we present a more detailed version of Algorithm~\ref{alg:sketch}, given
in Algorithm~\ref{alg:mineepisode}. The only missing part of the algorithm is the \textsc{AddIntermediate}
subroutine which we will give in the next section along with the proof of correctness.

\begin{algorithm}[htb!]
\capstart
\Input{a sequence $s$, a frequency threshold $\sigma$, a window size $\rho$}
\Out{$f$-closed frequent strict episodes}
	$\efam{C} \define$ all frequent episodes with one node\;
	$\efam{E} \define \emptyset$; $N \define 1$\;
	\While{there are episodes with $N$ or more nodes in $\efam{C} \cup \efam{E}$} {
		$M \define 0$\;
		\While{there are episodes with $M$ or more edges in $\efam{C} \cup \efam{E}$} {
			\ForEach{$G \in \efam{G}$} {
				\If{$\textsc{TestCandidate}(G, \efam{E})$ \AND $G$ is frequent} {
					add $G$ to $\efam{E}$\;
					compute and store $\iclosure{G}$\;
					add \textsc{AddIntermediate}$(G)$ to $\efam{E}$\;
				}
			}
			$\efam{P} \define \set{C \in \efam{E} \mid \abs{E(C)} = M, \abs{V(C)} = N}$\;
			$\efam{C} \define \textsc{GenerateCandidate}(\efam{P})$\; 
			$M \define M + 1$\;
		}
		$\efam{P} \define \set{C \in \efam{E} \mid \abs{V(G)} = N, G \text{ is a parallel episode}}$\;
		$\efam{C} \define \emptyset$\;
		\ForEach{$G \in \efam{P}$} {
			$x \define $ last node of $G$\;
			$\efam{H} \define \set{H \in \efam{G} \mid \abs{V(G) \cap V(H)} = n - 1, \ \lab{\text{last node of } H} > \lab{x}}$\;
			\ForEach{$H \in \efam{H}$} {
				add $G$ + last node of $H$ to $\efam{C}$\;
			}
			\tcpas{Check episode $G$ augmented with a node carrying the label of the last node $x$}
			$y \define$ a node with the same label as $x$\; 
			add $G + y + (x, y)$ to $\efam{C}$\;
		}
		$N \define N + 1$\;

	}
	\Return{\textsc{F-Closure}$(\set{\iclosure{G} \mid G \in \efam{E}})$}\;
\caption{\textsc{MineEpisodes}. An algorithm discovering all frequent closed episodes.}
\label{alg:mineepisode}
\end{algorithm}

\subsection{Proof of Correctness}
\label{sec:correct}

In the original framework for mining itemsets, the algorithm discovers itemset
generators, the smallest itemsets that produce the same closure. It can be shown
that generators form a downward closed collection, so they can be discovered
efficiently. This, however, is not true for episodes, as the following example demonstrates.

\begin{example}
\label{ex:fail}
Consider the episodes given in Figure~\ref{fig:fail} and a sequence $s =
acbxxabcxxbac$.  Assume that the window size is $\rho = 3$ and the
frequency threshold is $\sigma = 1$.  The frequency of $H$ is $1$ and there is no
subepisode that has the same frequency. Hence to discover this episode
all of its maximal subepisodes need to be discovered. One of these subepisodes
is $\iclosure{G}$. However, $\iclosure{G}$ is not added into $\efam{E}$ since
$(a, c) \in \iclosure{\iclosure{G} - (a, c)}$ and $\textsc{TestCandidate}$ returns false.
\end{example}

\begin{figure}[htb!]
\centering
\subfigure[$G$\label{fig:fail:a}]{\begin{minipage}[b]{2.3cm}\centering\normalsize\begin{tikzpicture}[>=latex',line join=bevel,scale=0.3]\pgfsetlinewidth{0.5bp}\pgfsetcolor{black}
  \draw [->] (53.916bp,71.831bp) .. controls (49.939bp,63.877bp) and (45.185bp,54.369bp)  .. (36.207bp,36.413bp);
\begin{scope}
  \definecolor{strokecol}{rgb}{0.0,0.0,0.0};
  \pgfsetstrokecolor{strokecol}
  \draw (63bp,90bp) node {$a$};
\end{scope}
\begin{scope}
  \definecolor{strokecol}{rgb}{0.0,0.0,0.0};
  \pgfsetstrokecolor{strokecol}
  \draw (99bp,18bp) node {$c$};
\end{scope}
\begin{scope}
  \definecolor{strokecol}{rgb}{0.0,0.0,0.0};
  \pgfsetstrokecolor{strokecol}
  \draw (27bp,18bp) node {$b$};
\end{scope}\end{tikzpicture}\end{minipage}}
\subfigure[$\iclosure{G}$]{\begin{minipage}[b]{2.3cm}\centering\normalsize\begin{tikzpicture}[>=latex',line join=bevel,scale=0.3]\pgfsetlinewidth{0.5bp}\pgfsetcolor{black}
  \draw [->] (53.916bp,71.831bp) .. controls (49.939bp,63.877bp) and (45.185bp,54.369bp)  .. (36.207bp,36.413bp);
  \draw [->] (72.084bp,71.831bp) .. controls (76.061bp,63.877bp) and (80.815bp,54.369bp)  .. (89.793bp,36.413bp);
\begin{scope}
  \definecolor{strokecol}{rgb}{0.0,0.0,0.0};
  \pgfsetstrokecolor{strokecol}
  \draw (63bp,90bp) node {$a$};
\end{scope}
\begin{scope}
  \definecolor{strokecol}{rgb}{0.0,0.0,0.0};
  \pgfsetstrokecolor{strokecol}
  \draw (99bp,18bp) node {$c$};
\end{scope}
\begin{scope}
  \definecolor{strokecol}{rgb}{0.0,0.0,0.0};
  \pgfsetstrokecolor{strokecol}
  \draw (27bp,18bp) node {$b$};
\end{scope}\end{tikzpicture}\end{minipage}}
\subfigure[$H$]{\begin{minipage}[b]{2.3cm}\centering\normalsize\begin{tikzpicture}[>=latex',line join=bevel,scale=0.3]\pgfsetlinewidth{0.5bp}\pgfsetcolor{black}
  \draw [->] (54bp,18bp) .. controls (56.615bp,18bp) and (59.229bp,18bp)  .. (71.93bp,18bp);
  \draw [->] (53.916bp,71.831bp) .. controls (49.939bp,63.877bp) and (45.185bp,54.369bp)  .. (36.207bp,36.413bp);
  \draw [->] (72.084bp,71.831bp) .. controls (76.061bp,63.877bp) and (80.815bp,54.369bp)  .. (89.793bp,36.413bp);
\begin{scope}
  \definecolor{strokecol}{rgb}{0.0,0.0,0.0};
  \pgfsetstrokecolor{strokecol}
  \draw (63bp,90bp) node {$a$};
\end{scope}
\begin{scope}
  \definecolor{strokecol}{rgb}{0.0,0.0,0.0};
  \pgfsetstrokecolor{strokecol}
  \draw (99bp,18bp) node {$c$};
\end{scope}
\begin{scope}
  \definecolor{strokecol}{rgb}{0.0,0.0,0.0};
  \pgfsetstrokecolor{strokecol}
  \draw (27bp,18bp) node {$b$};
\end{scope}\end{tikzpicture}\end{minipage}}
\caption{Toy episodes for Example~\ref{ex:fail}.}
\label{fig:fail}
\end{figure}
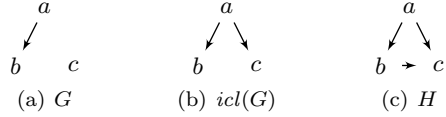

The core of the problem is that sometimes adding some particular edge is not
allowed until you have added another edge. For example, we cannot add $(b, c)$
to $G$ (given in Figure~\ref{fig:fail:a}) until we have added $(a, c)$.

We solve this problem by adding some additional episodes between the discovered
episode $G$ and its closure $\iclosure{G}$. Roughly speaking, given an episode
$G \in \sspace$ we need to add the episodes containing edges/nodes from $\iclosure{G}$ such
that these edges/nodes can be "hidden" by a single edge or node.

The details are given in the \textsc{AddIntermediate} algorithm given in
Algorithm~\ref{alg:addinter}, and the justification for adding these episodes
is given in the proof of Theorem~\ref{prop:discover}.

\begin{algorithm}[htb!]
\capstart
\Input{an episode $G \in \sspace$}
\Out{a subset of episodes between $G$ and $\iclosure{G}$}
	$\efam{C} \define \emptyset$\;
	\ForEach{$x \in V(\iclosure{G}) - V(G)$} {
		$H \define G + x$\;
		add $H$ to $\efam{C}$\;\nllabel{alg:future1}
		\ForEach{$y \in V(G)$ \AND $e = (x, y) \notin \iclosure{G}$} {
			$Z \define E(\closure{H + e}) - E(H + e)$\;
			\If{$Z \neq \emptyset$ \AND $Z \subset E(\iclosure{H})$} {
				add $H + Z$ to $\efam{C}$\;\nllabel{alg:future2}
			}
		}
		\ForEach{$y \in V(G)$ \AND $e = (y, x) \notin \iclosure{G}$} {
			$Z \define E(\closure{H + e}) - E(H + e)$\;
			\If{$Z \neq \emptyset$ \AND $Z \subset E(\iclosure{H})$} {
				add $H + Z$ to $\efam{C}$\;\nllabel{alg:future3}
			}
		}
	}
	\ForEach{$x, y \in V(\iclosure{G}) - V(G)$ \AND $e = (x, y) \notin \iclosure{G}$} {
			add $G + x + y$ to $\efam{C}$\;\nllabel{alg:future4}
	}

	\ForEach{$x, y \in V(G)$ \AND $e = (x, y) \notin \iclosure{G}$} {
		$Z \define E(\closure{G + e}) - E(G + e)$\;
		\If{$Z \neq \emptyset$ \AND $Z \subset E(\iclosure{G})$} {
			add $G + Z$ to $\efam{C}$\;\nllabel{alg:future5}
		}
	}
	\Return{$\efam{C}$}\;

\caption{\textsc{AddIntermediate}. An algorithm that given an episode adds some
additional episodes between the discovered episode and its closure.  This is
necessary to guarantee the correctness of the algorithm.}

\label{alg:addinter}
\end{algorithm}

We will now establish the correctness of the algorithm and justify Algorithm~\ref{alg:addinter}.

\begin{definition}
Let $G \in \sspace$ be an episode. We call a proper skeleton edge $e$ in $G$
\emph{derivable} if $G \preceq \iclosure{G - e}$.
We call a solitary node $v$ \emph{derivable} if $G \preceq \iclosure{G - v}$.
\end{definition}

We will need the following straightforward lemmae to prove our main result.

\begin{lemma}
\label{lem:removea}
Let $G \in \sspace$ be an episode. Let $e$ be  a non-derivable skeleton edge in $G$. Then
$e$ is a non-derivable skeleton edge in $G - f$, where $f \neq e$ is a proper skeleton edge
in $G$. Similarly,
$e$ is a non-derivable skeleton edge in $G - v$, where $v$ is a solitary node.
\end{lemma}

\begin{proof}
Assume that $e$ is derivable in $G - f$, then $G - f \preceq \iclosure{G - f -
e} \preceq \iclosure{G - e}$. On the other hand, $G - e \preceq \iclosure{G - e}$.
Since the closure operator adds nodes only with unique labels, the graph homomorpisms
from $G - e$ to $\iclosure{G - e}$ and from $G - f$ to $\iclosure{G - e}$ must be equal. 
This implies that $G \preceq \iclosure{G - e}$ which is a contradiction. A similar
argument holds for $G - v$.
\end{proof}

\begin{lemma}
\label{lem:removeb}
Let $G \in \sspace$ be an episode. Let $w$ be a non-derivable solitary node in $G$. Then
$w$ is a non-derivable node in $G - v$, where $v$ is a solitary node. Similarly,
$w$ is a non-derivable node in $G - e$, where $e$ is a proper skeleton edge in $G$.
\end{lemma}

\begin{proof}
Similar to the proof of Lemma~\ref{lem:removea}.
\end{proof}

\begin{lemma}
\label{lem:derivable}
Let $G \in \sspace$ be an episode with a derivable edge $e$.
Then $\iclosure{G} = \iclosure{G - e}$.
Let $G \in \sspace$ be an episode with a derivable node $v$.
Then $\iclosure{G} = \iclosure{G - v}$.
\end{lemma}

\begin{proof}
The monotonicity and idempotency properties imply $\iclosure{G} \preceq
\iclosure{\iclosure{G - e}} = \iclosure{G - e}$. Since $\iclosure{G - e} \preceq \iclosure{G}$, we have $\iclosure{G} = \iclosure{G - e}$.

To prove the second case, note that $\iclosure{G} \preceq \iclosure{\iclosure{G - v}} = \iclosure{G - v}$ which immediately implies that $\iclosure{G} = \iclosure{G - v}$.
\end{proof}

\begin{theorem}
\label{prop:discover}
Assume a frequent episode $G \in \sspace$ with no derivable edges or derivable nodes.
Then $G$ is discovered.
\end{theorem}

\begin{proof}
We will prove the theorem by induction. Obviously, the theorem holds for
episodes with a single node. Assume now that the theorem is true for all
subepisodes of $G$. Episode $G$ will be discovered if it passes the tests in
\textsc{TestCandidate}. To pass these tests, all episodes of form $G - e$,
where $e$ is a proper skeleton edge need to be discovered.  Assume that all
these subepisodes are discovered but we have $e \in E(\iclosure{G - e})$.  This
means, by definition, that $e$ is a derivable edge in $G$ which is a
contradiction. The same argument holds for episodes of form $G - v$, where
$v$ is a node.

Now assume that one of the subepisodes, say, $H$ is not discovered.
The induction assumption now implies that $H$ has either derivable nodes or edges.

\begin{lemma}
\label{lem:correct:a}
There is a subepisode $F \in \sspace$, $F \prec H$ with no derivable edges and nodes such that
$H \preceq \iclosure{F}$. Any non-derivable skeleton edge in $H$ will remain in $F$.
Any non-derivable solitary node in $H$ will remain in $F$.
\end{lemma}

\begin{proof}
Build a chain of episodes $H = H_1, H_2, \ldots, H_N = F$ such that $H_{i + 1}$
is obtained from $H_i$ by removing either a derivable node or a derivable edge
and $F$ has no derivable edges or nodes. Note that this sequence always exists
but may not necessarily be unique. Lemma~\ref{lem:derivable} implies that we
have $H_i \preceq \iclosure{H_i} = \iclosure{H_{i + 1}}$.  Idempotency and
monotonicity imply that $H \preceq \iclosure{F}$.  Lemma~\ref{lem:removea}
implies that if $e$ is a non-derivable skeleton edge in $H$, then $e$ is also a
non-derivable skeleton edge in each $H_i$.  Similarly, Lemma~\ref{lem:removeb}
implies that any non-derivable solitary node will remain in each $H_i$.
\end{proof}

By the induction assumption $F$ is discovered.  We claim that $H$ will be
discovered by \textsc{AddIntermediate}$(F)$.  Let us denote $Z = E(H) - E(F)$
and $W = V(H) - V(F)$. 

The next three lemmae describe different properties of $Z$ and $W$.

\begin{lemma}
\label{lem:case1}
Assume that $H = G - v$. Then $W  = \set{w}$ such that $\lab{v} = \lab{w}$ and $Z = \emptyset$.
\end{lemma}

\begin{proof}
Lemma~\ref{lem:removea} implies that all skeleton edges in $H$
are non-derivable. Lemma~\ref{lem:correct:a} implies that $Z = \emptyset$.
Removing $v$ can turn only one node, say $w$, into a solitary node.
This happens when $\lab{v} = \lab{w}$ and there are no other edges adjacent to $w$. 
\end{proof}

\begin{lemma}
\label{lem:case2}
Assume that $H = G - (x, y)$. Then $W \subseteq \set{x, y}$.
If $W = \set{x, y}$, then $Z = \emptyset$.
\end{lemma}

\begin{proof}
Let $z$ be a node in $G$ (and in $H$) such that $z \notin \set{x, y}$. If $z$
is a solitary node in $G$, it also a solitary node in $H$. Lemma~\ref{lem:correct:a} now implies that $z
\notin W$. If $z$ has an adjacent skeleton edge in $G$, say $f$, then $f \neq
(x, y)$.  Lemma~\ref{lem:removea} implies that $f$ is also a non-derivable skeleton edge in $H$. Lemma~\ref{lem:correct:a}
now implies that $z \notin W$.

If $W = \set{x, y}$, then $x$ (and $y$) cannot have any adjacent non-derivable
skeleton edges in $H$. Hence there are no edges, except for $(x, y)$, adjacent
to $x$ or $y$ in $G$. Lemma~\ref{lem:removea} implies that all skeleton edges in $H$
are non-derivable. Lemma~\ref{lem:correct:a} implies that $Z = \emptyset$.
\end{proof}

\begin{lemma}
\label{lem:case3}
Assume that $H = G - e$.  Then $Z = E(\closure{F + W + e)} - E(F + W + e) \subset E(\iclosure{F + W})$.
\end{lemma}

Note that $F + W + e$ is not necessarily transitively closed.

\begin{proof}
Write $F' = F + W$.
First note that $Z \subset E(H) \subseteq E(\iclosure{F'})$.
Lemmae~\ref{lem:removea}~and~\ref{lem:correct:a} imply that all skeleton edges
of $G$, except for $e$ are in $E(F) = E(F')$. Hence, we must have
$E(\closure{F' + e}) = E(G)$.

Also note that $Z \cup E(F' + e)  = E(H) + e = E(G)$.  Since $Z \cap E(F' +
e) = \emptyset$, we have $Z = E(G) - E(F' + e) = E(\closure{F' + e}) - E(F' + e)$.
\end{proof}

If $H = G - v$, then Lemma~\ref{lem:case1} implies that $H = F + w$ and we
discover $H$ on Line~\ref{alg:future1} during \textsc{AddIntermediate}$(F)$.
Assume that $H = G - e$. Write $(x, y) = e$. Lemma~\ref{lem:case2} now implies
that $W \subseteq \set{x, y}$ and Lemma~\ref{lem:case3} implies that $Z =
E(\closure{F + W + e)} - E(F + W + e)$.  We will show that there are 4
different possible cases:

\begin{enumerate}
\item $W = \set{x, y}$. Lemma~\ref{lem:case2} implies that $Z = \emptyset$.
This implies that $H = F + x + y$ and we discover $H$ on Line~\ref{alg:future4}
during \textsc{AddIntermediate}$(F)$.

\item $W = \set{w}$, where $w$ is either $x$ or $y$ and $Z = \emptyset$.  This
implies that $H = F + w$ and we discover $H$ on Line~\ref{alg:future1} during
\textsc{AddIntermediate}$(F)$.

\item $W = \emptyset$ and $Z \neq \emptyset$.
This implies that $H = F + Z$ and we discover $H$ on Line~\ref{alg:future5}
during \textsc{AddIntermediate}$(F)$.

\item $W = \set{w}$, where $w$ is either $x$ or $y$ and $Z \neq \emptyset$.
This implies that $H = F + w + Z$ and we discover $H$ either on Line~\ref{alg:future2} or on Line~\ref{alg:future3}
during \textsc{AddIntermediate}$(F)$.

\end{enumerate}
This compeletes the proof of the theorem.
\end{proof}

\begin{theorem}
\label{thr:correct}
Every frequent $i$-closed episode will be outputted.
\end{theorem}

\begin{proof}
\textsc{TestEpisode} will output $\iclosure{G}$ for each discovered episode
$G$.  Hence, we need to show that for each $i$-closed episode $H$ there is a discovered
episode $G$ such that $H = \iclosure{G}$.

We will prove the theorem by induction.  Let $H$ and $G$ be episodes such that
$H = \iclosure{G}$ and $H$ is an $i$-closed frequent episode. If $G$ contains
only one node, then $G$ will be discovered.  Assume that the theorem holds for
any episode $\iclosure{G'}$, where $G'$ is a subepisode of $G$. If $G$ has
derivable nodes or edges, then by Lemma~\ref{lem:derivable} there exists an
episode $G' \prec G$ such that $H = \iclosure{G} = \iclosure{G'}$ and so by the
induction assumption $G'$ is discovered, and $H$ is outputted. If $G$ has no derivable edges or nodes,
then Theorem~\ref{prop:discover} implies that $G$ is discovered.  This
completes the proof. 
\end{proof}

\section{Experiments}
\label{sec:experiments}

We tested our algorithm\footnote{The C++ implementation is given at \url{http://adrem.ua.ac.be/implementations/}} on three text datasets, \emph{address}, consisting of the
inaugural addresses by the presidents of the United States\footnote{taken from~\url{http://www.bartleby.com/124/pres68}}, merged to form
a single long sequence,
\emph{moby}, the novel Moby Dick by Herman Melville\footnote{taken
from \url{http://www.gutenberg.org/etext/15}}, and \emph{abstract}, consisting
of the first 739 NSF award abstracts from 1990\footnote{taken from \url{http://kdd.ics.uci.edu/databases/nsfabs/nsfawards.html}}, also merged into one long sequence.
We processed all three sequences using the
Porter Stemmer\footnote{\url{http://tartarus.org/~martin/PorterStemmer/}} and removed the
stop words. The characteristics of datasets are given in Table~\ref{tab:basic}.

In the implementation an episode graph was implemented using sparse notation:
the neighbourhood of a node was presented as a list of edges. To ensure
efficient scanning, the sequence was implemented as a set of linked lists, one
for each symbol.  The experiments were conducted on a computer with an AMD
Athlon 64 processor and 2GB memory.  The code was compiled with G++ 4.3.4.

\begin{table}[htb!]
\centering
\begin{tabular}{lrr}
\toprule
Sequence & Size & $\abs{\Sigma}$ \\
\midrule
\emph{moby}     & 105\,719 & 10\,277 \\
\emph{abstract} & 67\,828  & 6\,718  \\
\emph{address}  & 62\,066  & 5\,295  \\
\bottomrule
\end{tabular}
\caption{Characteristics of the sequences. The first column contains the size of the sequence, and the second column the number of unique symbols in the
sequence.}
\label{tab:basic}
\end{table}

We used a window of size $15$ for all our experiments and varied the frequency
threshold $\sigma$. The main goal of our experiments was to demonstrate how we
tackle the problem of pattern explosion. Tables~\ref{tab:exp:f}
and~\ref{tab:exp:m} show how the total number of frequent episodes compared
with the identified $i$-closed, $e$-closed and $f$-closed episodes we
discovered in the three datasets, using the fixed-window frequency and the
disjoint-window frequency, respectively. The results show that while the reduction is
relatively small for large thresholds, its benefits are clearly visible as we lower the threshold.
The reason for this is that, because of the data characteristics,
the major part of the output consists of episodes
with a small number of nodes. Such episodes tend to be closed. When the threshold
is lowered, episodes with a large number of nodes become frequent which
leads to a pattern explosion.
In the extreme case, we ran out of memory when discovering all frequent episodes
for certain low thresholds while were able to compute the $f$-closed episodes.

\begin{table}[htb!]
\centering
\begin{tabular}{lrrrrr}
\toprule
Dataset & $\sigma$ & $f$-closed & $i$-closed & $e$-closed & frequent\\
\midrule
\emph{address} & $200$ & $1\,983$ & $1\,989$ & $1\,989$ & $1\,992$\\
                 & $100$ & $6\,774$ & $6\,817$ & $6\,820$ & $6\,880$\\
                 & $50$ & $25\,732$ & $26\,078$ & $26\,212$ & $34\,917$\\
                 & $30$ & $70\,820$ & $72\,570$ & $73\,328$ & $119\,326$\\
                 & $20$ & $166\,737$ & $192\,544$ & $207\,153$ & out of memory\\
\midrule
\emph{abstract} & $500$ & $893$ & $914$ & $916$ & $933$\\
                 & $400$ & $1\,374$ & $1\,436$ & $1\,440$ & $1\,499$\\
                 & $300$ & $2\,448$ & $2\,802$ & $2\,903$ & $4\,080$\\
                 & $200$ & $6\,074$ & $7\,374$ & $8\,080$ & $98\,350$\\
                 & $100$ & $26\,140$ & $43\,020$ & $95\,811$ & out of memory\\
\midrule
\emph{moby} & $200$ & $3\,389$ & $3\,394$ & $3\,394$ & $3\,395$\\
            & $100$ & $11\,018$ & $11\,079$ & $11\,084$ & $11\,127$\\
            & $50$ & $37\,551$ & $38\,043$ & $38\,120$ & $39\,182$\\
            & $30$ & $99\,937$ & $102\,380$ & $103\,075$ & $151\,115$\\
            & $20$ & $231\,563$ & $245\,683$ & $253\,208$ & out of memory\\

\bottomrule
\end{tabular}
\caption{The number of frequent, $i$-closed, $e$-closed and $f$-closed episodes with varying fixed-window frequency thresholds for the \emph{address}, \emph{abstract} and \emph{moby} datasets, respectively.}
\label{tab:exp:f}
\end{table}

\begin{table}[htb!]
\centering
\begin{tabular}{lrrrrr}
\toprule
Dataset & $\sigma$ & $f$-closed & $i$-closed & $e$-closed & frequent\\
\midrule
\emph{address} & $20$ & $2\,264$ & $2\,282$ & $2\,282$ & $2\,291$\\
                 & $10$ & $9\,984$ & $10\,213$ & $10\,219$ & $10\,396$\\
                 & $5$ & $46\,902$ & $50\,634$ & $50\,920$ & $65\,139$\\
                 & $4$ & $77\,853$ & $87\,935$ & $89\,076$ & $268\,675$\\
                 & $3$ & $149\,851$ & $187\,091$ & $195\,379$ & out of memory\\
\midrule
\emph{abstract} & $100$ & $195$ & $195$ & $195$ & $195$\\
                 & $50$ & $932$ & $1\,000$ & $1\,002$ & $1\,020$\\
                 & $40$ & $1\,677$ & $2\,053$ & $2\,122$ & $2\,585$\\
                 & $30$ & $3\,542$ & $5\,482$ & $5\,798$ & $20\,314$\\
                 & $20$ & $9\,933$ & $22\,945$ & $61\,513$ & out of memory\\
\midrule
\emph{moby} & $20$ & $4\,076$ & $4\,119$ & $4\,121$ & $4\,121$\\
            & $10$ & $15\,930$ & $16\,325$ & $16\,341$ & $16\,468$\\
            & $5$ & $67\,180$ & $72\,306$ & $72\,468$ & $77\,701$\\
            & $4$ & $109\,572$ & $122\,423$ & $122\,919$ & $141\,940$\\
            & $3$ & $207\,031$ & $251\,757$ & $158\,303$ & out of memory\\

\bottomrule
\end{tabular}
\caption{The number of frequent, $i$-closed, $e$-closed and $f$-closed episodes with varying disjoint-window frequency thresholds for the \emph{address}, \emph{abstract} and \emph{moby} datasets, respectively.}
\label{tab:exp:m}
\end{table}

To get a more detailed picture we examined the ratio of the number of frequent
episodes and the number of $f$-closed episodes, the ratio of the number of
$i$-closed episodes and the number of $f$-closed episodes, and the ratio of the
number of $e$-closed episodes and the number of $i$-closed episodes, as a
function of the number of nodes. The results using the fixed-window frequency
are shown in Figure~\ref{fig:ratio}.  We see that while there is
no improvement with small episodes, using closed episodes is essential if we
are interested in large episodes. In such a case we were able to reduce the
output by several orders of magnitude. For example, in the \emph{moby} dataset,
with a threshold of 30, there were $24\,131$ frequent episodes of size 6, of
which only 13 were $f$-closed. Clearly, the number of discovered $i$-closed
episodes remains greater than the number of $f$-closed episodes, but does not
explode, guaranteeing the feasibility of our algorithm. For example, in the
\emph{abstract} dataset, with a threshold of 200, there were $17\,587$ frequent
episodes of size 5, of which 878 were $i$-closed and 289 $f$-closed. Furthermore,
we can see that while using only the $e$-closure helps reduce the output
significantly, using the $i$-closure gives an even better reduction. For
example, in the \emph{address} dataset, with a threshold of 30, there were 259
$e$-closed episodes of size 5, of which 178 were $i$-closed.

The same ratios using the disjoint-window frequency are shown in Figure~\ref{fig:ratiomw}. Again, we can clearly see the benefits of using $i$-closure, especially on large episodes.

\begin{figure}[htb!]
\centering
\subfigure[\label{fig:ratio:a}Frequent / $f$-closed]{
\begin{tikzpicture} 
\begin{semilogyaxis}[xlabel=\# of nodes in episodes, ylabel=frequent / $f$-closed,
    width = 4.2cm,
    xtick = {1, ..., 7},
	ytick = {1, 10, 100, 1000, 10000, 100000},
	ymax = 100000,
	yminorticks = false,
	cycle list name=yaf,
	legend pos = north west
    ]

\addplot coordinates { 
(1, 1.002039033)
(2, 1.048790918)
(3, 1.471707053)
(4, 6.734597156)
(5, 108.8350515)
(6, 1406.6)
(7, 714)
};

\addplot coordinates { 
(1, 1.0019805)
(2, 1.047665601)
(3, 1.395519468)
(4, 5.266156463)
(5, 77.86861314)
(6, 1856.230769)
};

\addplot coordinates { 
(1, 1)
(2, 1.0204285)
(3, 1.611629183)
(4, 7.400614754)
(5, 60.85467128)
(6, 1429.083333)
(7, 11101.33333)
}; 
\legend{addr., moby, abs.}

\pgfplotsextra{\yafdrawaxis{1}{7}{1}{100000}}
\end{semilogyaxis} 
\end{tikzpicture}} 
\subfigure[\label{fig:ratio:b}$i$-closed / $f$-closed]{
\begin{tikzpicture} 
\begin{axis}[xlabel=\# of nodes in episodes, ylabel=$i$-closed / $f$-closed,
    width = 4.2cm,
    xtick = {1, ..., 7},
    ytick = {1.0, 1.5, ..., 4.0},
	y tick label style = {/pgf/number format/fixed zerofill},
    yticklabel = {$\pgfmathprintnumber[precision=1]{\tick}$},
	cycle list name=yaf,
    ]

\addplot coordinates { 
(1, 1.002039033)
(2, 1.002041054)
(3, 1.073507513)
(4, 1.432227488)
(5, 1.835051546)
(6, 1.866666667)
(7, 2)
};

\addplot coordinates { 
(1, 1.0019805)
(2, 1.002583691)
(3, 1.077382117)
(4, 1.505102041)
(5, 1.481751825)
(6, 1.230769231)
};

\addplot coordinates { 
(1, 1)
(2, 1.002491281)
(3, 1.076247943)
(4, 1.506147541)
(5, 3.038062284)
(6, 4)
(7, 1.333333333)
}; 

\pgfplotsextra{\yafdrawaxis{1}{7}{1}{4}}
\end{axis} 
\end{tikzpicture}} 
\subfigure[\label{fig:ratio:c}$e$-closed / $i$-closed]{
\begin{tikzpicture} 
\begin{axis}[xlabel=\# of nodes in episodes, ylabel=$e$-closed / $i$-closed,
    width = 4.2cm,
    xtick = {1, ..., 7},
    ytick = {1.0, 1.1, ..., 1.5},
	y tick label style = {/pgf/number format/fixed zerofill},
    yticklabel = {$\pgfmathprintnumber[precision=1]{\tick}$},
	cycle list name=yaf,
    ]

\addplot coordinates { 
(1, 1)
(2, 1.003394828)
(3, 1.018221942)
(4, 1.133024487)
(5, 1.45505618)
(6, 1.428571429)
(7, 1)
};

\addplot coordinates { 
(1, 1)
(2, 1.003566056)
(3, 1.012803051)
(4, 1.076836158)
(5, 1.103448276)
(6, 1)
};

\addplot coordinates { 
(1, 1)
(2, 1.008449304)
(3, 1.109582059)
(4, 1.24829932)
(5, 1.109339408)
(6, 1.135416667)
(7, 1)
};

\pgfplotsextra{\yafdrawaxis{1}{7}{1}{1.45}}
\end{axis} 
\end{tikzpicture}}

\caption{Ratios of episodes as a function of the number of events.
$\freqf{G}$ was used as frequency with the threshold $\sigma = 30$ for \emph{address} and \emph{moby},
and $\sigma = 200$ for \emph{abstract}.  Note that the y-axis of Figure~\ref{fig:ratio:a} is in log-scale.}
\label{fig:ratio}
\end{figure}

\begin{figure}[htb!]
\centering
\subfigure[\label{fig:ratiomw:a}Frequent / $f$-closed]{
\begin{tikzpicture} 
\begin{semilogyaxis}[xlabel=\# of nodes in episodes, ylabel=frequent / $f$-closed,
    width = 4.2cm,
    xtick = {1, ..., 7},
	ytick = {1, 10, 100, 1000, 10000, 100000},
	ymax = 100000,
	yminorticks = false,
	cycle list name=yaf,
	legend pos = north west
    ]

\addplot coordinates { 
(1, 1.010912265)
(2, 1.076568487)
(3, 1.654540077)
(4, 7.252204586)
(5, 228.5915493)
(6, 5084.333333)
(7, 49463.5)
};

\addplot coordinates { 
(1, 1.017654172)
(2, 1.077832783)
(3, 1.553494554)
(4, 4.732283465)
(5, 18.76666667)
(6, 1)
};

\addplot coordinates { 
(1, 1)
(2, 1.023550725)
(3, 1.276315789)
(4, 3.804761905)
(5, 59.6744186)
(6, 1600.5)
}; 
\legend{addr., moby, abs.}

\pgfplotsextra{\yafdrawaxis{1}{7}{1}{100000}}
\end{semilogyaxis} 
\end{tikzpicture}} 
\subfigure[\label{fig:ratiomw:b}$i$-closed / $f$-closed]{
\begin{tikzpicture} 
\begin{axis}[xlabel=\# of nodes in episodes, ylabel=$i$-closed / $f$-closed,
    width = 4.2cm,
    xtick = {1, ..., 7},
	ymax = 34,
	ytick = {1, 5, 10, 15, 20, 25, 30, 34},
	cycle list name=yaf,
    ]

\addplot coordinates { 
(1, 1.010912265)
(2, 1.041087963)
(3, 1.247855664)
(4, 2.088183422)
(5, 4.169014085)
(6, 6.222222222)
(7, 5.5)
};

\addplot coordinates { 
(1, 1.017654172)
(2, 1.04338197)
(3, 1.218658808)
(4, 2.181102362)
(5, 3.922222222)
(6, 1)
};

\addplot coordinates { 
(1, 1)
(2, 1.020833333)
(3, 1.185855263)
(4, 2.150793651)
(5, 9.988372093)
(6, 33.16666667)
}; 

\pgfplotsextra{\yafdrawaxis{1}{7}{1}{34}}
\end{axis} 
\end{tikzpicture}} 
\subfigure[\label{fig:ratiomw:c}$e$-closed / $i$-closed]{
\begin{tikzpicture} 
\begin{axis}[xlabel=\# of nodes in episodes, ylabel=$e$-closed / $i$-closed,
    width = 4.2cm,
    xtick = {1, ..., 7},
    ytick = {1.0, 1.1, ..., 1.71},
	ymax = 1.7,
	y tick label style = {/pgf/number format/fixed zerofill},
    yticklabel = {$\pgfmathprintnumber[precision=1]{\tick}$},
	cycle list name=yaf,
    ]

\addplot coordinates { 
(1, 1)
(2, 1.002951831)
(3, 1.014797007)
(4, 1.09375)
(5, 1.601351351)
(6, 1.696428571)
(7, 1)
};

\addplot coordinates { 
(1, 1)
(2, 1.002871769)
(3, 1.005690035)
(4, 1.013607331)
(5, 1)
(6, 1)
};

\addplot coordinates { 
(1, 1)
(2, 1.000887311)
(3, 1.025658807)
(4, 1.13800738)
(5, 1.105937136)
(6, 1)
};

\pgfplotsextra{\yafdrawaxis{1}{7}{1}{1.7}}
\end{axis} 
\end{tikzpicture}}

\caption{Ratios of episodes as a function of the number of events.
$\freqm{G}$ was used as frequency with the threshold $\sigma = 4$ for \emph{address} and \emph{Moby},
and $\sigma = 30$ for \emph{abstract}.  Note that the y-axis of Figure~\ref{fig:ratiomw:a} is in log-scale.}
\label{fig:ratiomw}
\end{figure}

The difference between the number of $f$-closed and $i$-closed episodes can be explained by the fact that the $i$-closure operator looks at \emph{all}
valid mappings of the episode while the coverage requires only one valid mapping to exist. For example, consider the episode
given in Example~\ref{ex:intro1}. This episode is $f$-closed with respect to disjoint windows, and occurs in 10 windows. A subepisode
\[
	\text{chief} \to \text{justic} \qquad \text{vice} \qquad \text{president}
\]
is $i$-closed but not $f$-closed. The reason for this is that several speeches contain
a line 'vice president \ldots chief justice \ldots vice president'. Hence, we can construct a valid
mapping for the order: president, chief, justice, vice. Consequently, we cannot add an edge
from vice to president. However, for all such mappings we can construct an alternative mapping
satisfying the episode in Example~\ref{ex:intro1}. Thus the support of both episodes will be the same
and we can delete the subepisode from the output.

To complete our analysis, we present a comparison of the number of serial, parallel and general episodes we found.
We compare the number of $f$-closed episodes to the overall number of frequent episodes, to illustrate how much
the output has been reduced by using closed episodes. The results are shown in Tables~\ref{tab:ep} and~\ref{tab:ep:m} for fixed-window and disjoint-window frequency, respectively.
To avoid double counting, we consider singletons to be parallel episodes, while neither serial nor parallel episodes are included in the general episodes total.
As expected, the number of serial and parallel episodes does not change much, as
most of them are closed. For a serial episode not to be closed, we would need to find an episode consisting of more nodes, yet having the same frequency,
which is not often the case. A parallel episode is not closed if a partial order can be imposed on its nodes, without a decline in frequency --- again, this is not often the case.
However, as has been pointed out in Figure~\ref{fig:explosion} in the introduction, a single frequent serial episode results in an explosion of the number of discovered general episodes.
The results demonstrate that the output of general episodes has indeed been greatly reduced. For example, using a fixed-window frequency threshold of $200$ on the \emph{abstract} dataset, we discovered $93\,813$ frequent general episodes, of which only $2\,247$ were $f$-closed.

\begin{table}[htb!]
\centering
\begin{tabular}{lrrrrrrr}
\toprule
& & serial & serial & parallel & parallel & general & general\\
Dataset & $\sigma$ & $f$-closed & frequent & $f$-closed & frequent & $f$-closed & frequent\\
\midrule
\emph{address} & $200$ & $293$ & $293$ & $1\,670$ & $1\,674$ & $20$ & $25$\\
                 & $100$ & $1\,739$ & $1\,742$ & $4\,846$ & $4\,878$ & $189$ & $260$\\
                 & $50$ & $8\,436$ & $8\,494$ & $15\,526$ & $16\,068$ & $1\,770$ & $10\,355$\\
                 & $30$ & $24\,620$ & $24\,973$ & $36\,593$ & $41\,055$ & $9\,607$ & $53\,298$\\
                 & $20$ & $61\,254$ & N/A & $67\,658$ & N/A & $37\,825$ & N/A\\
\midrule
\emph{abstract} & $500$ & $116$ & $116$ & $667$ & $669$ & $110$ & $148$\\
                 & $400$ & $206$ & $208$ & $939$ & $942$ & $229$ & $349$\\
                 & $300$ & $433$ & $448$ & $1\,435$ & $1\,471$ & $580$ & $2\,161$\\
                 & $200$ & $1\,124$ & $1\,353$ & $2\,703$ & $3\,184$ & $2\,247$ & $93\,813$\\
                 & $100$ & $4\,597$ & N/A & $7\,854$ & N/A & $13\,689$ & N/A\\
\midrule
\emph{moby} & $200$ & $594$ & $594$ & $2\,772$ & $2\,776$ & $23$ & $25$\\
            & $100$ & $2\,992$ & $2\,997$ & $7\,749$ & $7\,788$ & $277$ & $342$\\
            & $50$ & $12\,529$ & $12\,583$ & $22\,615$ & $23\,192$ & $2\,407$ & $3\,407$\\
            & $30$ & $34\,469$ & $34\,915$ & $52\,481$ & $58\,021$ & $12\,987$ & $58\,179$\\
            & $20$ & $85\,112$ & N/A & $96\,675$ & N/A & $49\,866$ & N/A\\

\bottomrule
\end{tabular}
\caption{The number of $f$-closed and frequent serial, parallel and general episodes, with varying fixed-window frequency thresholds for the \emph{address}, \emph{abstract} and \emph{moby} datasets, respectively. Singletons are classified as parallel episodes, and general episodes do not include the serial and parallel episodes.}
\label{tab:ep}
\end{table}

\begin{table}[htb!]
\centering
\begin{tabular}{lrrrrrrr}
\toprule
& & serial & serial & parallel & parallel & general & general\\
Dataset & $\sigma$ & $f$-closed & frequent & $f$-closed & frequent & $f$-closed & frequent\\
\midrule
\emph{address} & $20$ & $479$ & $479$ & $1\,744$ & $1\,753$ & $41$ & $59$\\
                 & $10$ & $3\,004$ & $3\,012$ & $6\,325$ & $6\,468$ & $655$ & $916$\\
                 & $5$ & $15\,318$ & $15\,557$ & $24\,038$ & $27\,069$ & $7\,546$ & $22\,513$\\
                 & $4$ & $25\,532$ & $26\,372$ & $36\,116$ & $45\,108$ & $16\,205$ & $197\,195$\\
                 & $3$ & $49\,859$ & N/A & $57\,293$ & N/A & $42\,699$ & N/A\\
\midrule
\emph{abstract} & $100$ & $20$ & $20$ & $160$ & $160$ & $15$ & $15$\\
                 & $50$ & $164$ & $166$ & $565$ & $573$ & $203$ & $281$\\
                 & $40$ & $300$ & $314$ & $892$ & $924$ & $485$ & $1\,347$\\
                 & $30$ & $631$ & $708$ & $1\,479$ & $1\,660$ & $1\,432$ & $17\,946$\\
                 & $20$ & $1\,638$ & N/A & $2\,920$ & N/A & $5\,375$ & N/A\\
\midrule
\emph{moby} & $20$ & $1\,010$ & $1\,014$ & $2\,983$ & $3\,004$ & $83$ & $103$\\
            & $10$ & $5\,075$ & $5\,099$ & $9\,853$ & $10\,035$ & $1\,002$ & $1\,334$\\
            & $5$ & $22\,016$ & $22\,304$ & $34\,070$ & $37\,920$ & $11\,094$ & $17\,477$\\
            & $4$ & $35\,950$ & $36\,807$ & $50\,545$ & $61\,704$ & $23\,077$ & $43\,429$\\
            & $3$ & $69\,048$ & N/A & $79\,944$ & N/A & $58\,039$ & N/A\\

\bottomrule
\end{tabular}
\caption{The number of $f$-closed and frequent serial, parallel and general episodes, with varying disjoint-window frequency thresholds for the \emph{address}, \emph{abstract} and \emph{moby} datasets, respectively. Singletons are classified as parallel episodes, and general episodes do not include the serial and parallel episodes.}
\label{tab:ep:m}
\end{table}

The runtimes of our experiments varied between a few seconds and $3$ minutes for
the largest experiments. However, with low thresholds, our algorithm for
finding closed episodes ran faster than the algorithm for finding all frequent
episodes, and at the very lowest thresholds, our algorithm produced results,
while the frequent-episodes algorithm ran out of memory. This demonstrates the
infeasibility of approaching the problem by first generating all frequent
episodes, and then pruning the non-closed ones. The $i$-closed episodes are the
necessary intermediate step.

\section{Related Work}
\label{sec:related}
Searching for frequent patterns in data is a very common data mining
problem. The first attempt at
discovering sequential patterns was made by Wang et
al.~\cite{wang:94:combinatorial}. There, the dataset consists of a number of
sequences, and a pattern is considered interesting if it is long enough and can
be found in a sufficient number of sequences. The method proposed in this
paper, however, was not guaranteed to discover all interesting patterns, but a
complete solution to a more general problem (dropping the pattern length
constraint) was later provided by Agrawal and Srikant~\cite{agrawal:95:mining}
using an \textsc{Apriori}-style algorithm~\cite{agrawal:94:fast}.

It has been argued that not all discovered patterns are of interest to the
user, and some research has gone into outputting only closed sequential
patterns, where a sequence is considered closed if it is not properly contained
in any other sequence which has the same frequency. Yan et
al.~\cite{yan:03:clospan}, Tzvetkov et al.~\cite{tzvetkov:03:mining}, and Wang
and Han~\cite{wang:04:bide} proposed methods for mining such closed patterns,
while Garriga~\cite{garriga:05:summarizing} further reduced the output by
post-processing it and representing the patterns using partial orders.  Despite
their name, the patterns discovered by Garriga are different from the
traditional episodes. A sequence covers an episode if \emph{every} node of
the DAG can be mapped to a symbol such that the order is respected, whereas a
partial order discovered by Garriga is covered by a sequence if all paths in the DAG occur in the sequence; however, 
a single event in a sequence can be mapped to multiple nodes.

In another attempt to trim the output, Garofalakis et al.~\cite{garofalakis:02:mining}
proposed a family of algorithms called \textsc{Spirit} which allow the user to
define regular expressions that specify the language that the discovered patterns must
belong to.

Looking for frequent episodes in a single event sequence was first proposed by
Mannila et al.~\cite{mannila:97:discovery}. The \textsc{Winepi} algorithm finds
all episodes that occur in a sufficient number of windows of fixed length. The
frequency of an episode is defined as the fraction of all fixed-width sliding
windows in which the episode occurs. The user is required to choose the width
of the window and a frequency threshold. Specific algorithms are given for the
case of parallel and serial episodes. However, no algorithm for detecting
general episodes (DAGs) is provided.

The same paper proposes the \textsc{Minepi} method, where the interestingness of an episode is measured by the number of minimal windows that contain it. As was shown by Tatti~\cite{tatti:09:significance}, \textsc{Minepi} fails due to an error in its definition. Zhou et al.~\cite{zhou:10:mining} proposed mining closed serial episodes based on the \textsc{Minepi} method, without solving this error. Laxman et al. introduced a monotonic measure as the
maximal number of non-overlapping occurrences of the
episode~\cite{laxman:07:fast}.

Pei et al.~\cite{pei:06:discovering} considered a restricted version of our
problem setup. In their setup, items are allowed to occur only once in a window
(string in their terminology). This means that the
discovered episodes can contain only one occurrence of each item. This
restriction allows them to easily construct closed episodes. Our setup is more general
since we do not restrict the number of occurrences of a symbol in the window and
the miner introduced by Pei cannot be adapted to our problem setting
since the restriction imposed by the authors plays a vital part in their algorithm.

Garriga~\cite{casas-garriga:03:discovering} pointed out that
\textsc{Winepi} suffers from bias against longer episodes, and proposed 
solving this by increasing the window length proportionally to the episode length. However, as
was pointed out by M\'eger and Rigotti \cite{meger:04:constraint-based}, the
algorithm given in this paper contained an error.

An attempt to define frequency without using any windows has been made by
Calders et al.~\cite{calders:07:mining} where the authors define an
interestingness measure of an itemset in a stream to be the frequency
starting from a point in time that maximizes it. However, this method is defined only for itemsets, or parallel episodes,
and not for general
episodes. Cule et al.~\cite{cule:09:new} proposed a
method that uses neither a window of fixed size, nor minimal occurrences, and
an interestingness measure is defined as a combination of the cohesion and the
frequency of an episode --- again, only for parallel episodes. Tatti~\cite{tatti:09:significance} and Gwadera et
al.~\cite{gwadera:05:reliable,gwadera:05:markov} define
an episode as interesting if its occurrences deviate from expectations.

Finally, an extensive overview of temporal data mining has been made by Laxman
and Sastry~\cite{laxman:06:survey}.

\section{Conclusions}
\label{sec:conclusions}

In this paper, we tackled the problem of pattern explosion when mining frequent
episodes in an event sequence. In such a setting, much of the output is
redundant, as many episodes have the same frequency as some other, more
specific, episodes. We therefore output only closed episodes, for which this is
not the case. Further redundancy is found in the fact that some episodes can be
represented in more than one way. We solve this problem by restricting
ourselves to strict, transitively closed episodes.

Defining frequency-closed episodes created new problems, as, unlike in some
other settings, a non-closed frequent episode can have more than one closure.
To solve this, we defined a closure operator based on instances.
This closure does not suffer from the same problems that occur with the closure based
on frequency. Unlike
the closure based on frequency, an episode will always have only one instance-closure.

We further proved that every
$f$-closed episode must also be $i$-closed. Based on this, we developed an
algorithm that efficiently identifies $i$-closed episodes, as well as $f$-closed
episodes, in a post-processing step. Experiments have confirmed that the
reduction in output is considerable, and essential for large episodes, where we
reduced the output by several orders of magnitude. Moreover, thanks to introducing $i$-closed episodes, we can now produce output for thresholds at which finding all frequent episodes is infeasible.

\section*{Acknowledgments}
Nikolaj Tatti is supported by a Post-doctoral Fellowship of the Research Foundation Flanders (\textsc{fwo}).

\bibliographystyle{plain}
\bibliography{bibliography}

\end{document}